\documentclass[a4paper,10pt]{amsart}
\usepackage[english]{babel}
\usepackage{a4wide}
\usepackage[utf8]{inputenc}
\usepackage{amssymb,amsmath,amscd,amsfonts,amsthm,bbm,mathrsfs,enumerate}
\usepackage{booktabs}
\usepackage[colorlinks=true, linkcolor=blue, citecolor=blue]{hyperref}
\usepackage{mathabx}
\usepackage{graphicx}
\usepackage{color}
\usepackage{accents}
\usepackage{subcaption}
\usepackage{enumitem}

\DeclareUnicodeCharacter{00A0}{~}

\theoremstyle{definition}
\newtheorem{theorem}{Theorem}
\newtheorem*{theorem*}{Statement}
\newtheorem{lemma}{Lemma}

\newtheorem{proposition}[lemma]{Proposition}

\newtheorem{remark}{Remark}

\newtheorem*{condition*}{Condition}
\newtheorem{assumption}{Assumption}

\DeclareMathOperator{\support}{supp}

\DeclareMathOperator{\LCP}{LCP}
\DeclareMathOperator{\PL}{PL}

\newcommand{\1}{\mathbbm 1}
\newcommand{\T}{\top}

\newcommand{\PP}{{{\mathbb P}}} 
\newcommand{\EE}{{{\mathbb E}}} 
 
\newcommand{\RR}{{{\mathbb R}}}

\newcommand{\mcL}{{\mathscr L}}

\newcommand{\cL}{{\mathcal L}} 
 
\newcommand{\cN}{{\mathcal N}} 
 
\newcommand{\cP}{{\mathcal P}}

\newcommand{\oZ}{{\bar Z}}

\newcommand{\ps}[1]{\langle #1 \rangle}

\newcommand{\bs}{\boldsymbol}
\newcommand{\eqlaw}{\stackrel{\cL}{=}}  
\newcommand{\toaslong}{\xrightarrow[N\to\infty]{\text{a.s.}}}

\title[Equilibria of large random Lotka-Volterra systems]{Equilibria of large random Lotka-Volterra systems \\ 
with vanishing species:  
a mathematical approach}

\author[Akjouj et al.]{Imane Akjouj$^{(1)}$, Walid Hachem$^{(2)}$,\\ Myl\`ene Ma\"ida$^{(1),(3)}$, Jamal Najim$^{(2)}$}

\date{\today\\
{\small (1)} Univ. Lille, CNRS, UMR 8524 - Laboratoire Paul Painlevé, F-59000 Lille, France.\\
{\small (2)} IRL CRM-CNRS, International Reserach Lab, Montréal, Canada\\
{\small (3)} Université Gustave Eiffel and CNRS, France. }
\begin{document}

\maketitle

\begin{abstract} 
Ecosystems with a large number of species are often modelled as
Lotka-Volterra dynamical systems built around a large random interaction
matrix. Under some known conditions, a global equilibrium exists and is unique. 
In this article, we rigorously study its statistical properties in the large dimensional regime. 
Such an equilibrium vector is known to be the
solution of a so-called Linear Complementarity Problem (LCP). 
We describe its statistical properties by designing an Approximate Message Passing (AMP) algorithm, a technique that has
recently aroused an intense research effort in the fields of statistical
physics, Machine Learning, or communication theory.  Interaction matrices taken
from the Gaussian Orthogonal Ensemble, or following a Wishart distribution are
considered. Beyond these models, the AMP approach developed in this article has
the potential to describe the statistical properties of equilibria associated to more involved interaction matrix models.
\end{abstract} 

\setcounter{tocdepth}{1}

\section{Introduction}
\label{intro}

\subsection*{Equilibrium of a large Lotka-Volterra system}
In the field of mathematical ecology, Lotka-Volterra (LV) systems of coupled 
differential equations are widely used to model the time evolution of the
abundances of $N$ interacting species within an ecosystem \cite{tak-livre96}.
Such systems take the form 
\begin{equation}\label{eq:LV-system}
\frac{d\bs{x}_N}{dt} (t) = \bs{x}_N(t) \odot 
  \left( \bs{r}_N - \left(I_N - \Sigma_N \right) \bs{x}_N(t) \right), 
  \quad 
 \bs{x}_N(0) \in (0,\infty)^N ,  
\end{equation}
where the vector function $\bs{x}_N : [0,\infty) \to \RR_+^N = [0,\infty)^N$
represents the abundances of the $N$ species, $\odot$ is the componentwise 
product, $\bs{r}_N \in \RR_+^N$ is the so-called vector of intrinsic growth rates of the species, and 
$\Sigma_N =(\Sigma_{ij}) \in \RR^{N\times N}$ represents the interaction matrix. More precisely $\Sigma_{ij}$ represents the effect of species $j$ on the growth of species $i$. Equivalently, \eqref{eq:LV-system} can be written as a series of coupled ordinary differential equations:
$$
\frac{dx_i}{dt} (t) = x_i(t) \left( r_i - x_i(t) +\left(\Sigma_N\bs{x}_N\right)_i(t) \right)\ ,\quad x_i(0)>0\ ,\quad 1\le i\le N\,,
$$
where $\bs{x}_N=(x_i)$ and $\bs{r}_N=(r_i)$.

In theoretical ecology, the interaction matrix $\Sigma_N$ and the vector $\bs{r}_N$ are often modelled as random when the number $N$ of species is large, turning 
the ecological system into a large disordered system.
Such systems have aroused an important
amount of research in the fields of mathematical ecology, borrowing
tools from statistical physics, high dimensional probability, or random matrix
theory \cite{akj-etal-(arxiv)22}.  

In this paper, we shall be interested in the situation where the LV dynamical
system is well-defined for all $t\in \RR_+$ and
possesses an unique globally stable equilibrium vector: $$
\bs{x}_N^\star = \left( x^{\star}_i\right)_{i=1}^N\qquad \textrm{with}\qquad \bs{x}_N(t) \xrightarrow[t\to \infty]{} \bs{x}^\star_N\, .
$$
It is well-known that the property $\bs{x}_N(0)\in (0,\infty)^N$ is maintained for all $t>0$ and $\bs{x}_N(t)\in (0,\infty)^N$. However, in general, the
equilibrium vector $\bs{x}_N^\star$ may lie at the boundary of $\RR_+^N$, i.e. may have vanishing components. 
Moreover, assuming that $\Sigma_N$ and $\bs{r}_N$ are random, the vector $\bs{x}_N^\star$ is random as well.

When $N$ becomes large, it is of interest to understand the statistical properties of $\bs{x}^\star_N$ such as for example its proportion of non-zero components, or the distribution of $\bs{x}^\star_N$'s components, encoded in the empirical probability measure
$$
\mu^{\bs{x}_N^\star} = \frac 1N \sum_{i=1}^N
 \delta_{x^{\star}_i}\ ,
$$ 
where $ \delta_a$ stands for the Dirac measure at $a$.
Measure $\mu^{\bs{x}_N^\star}$ is a random measure defined on the
probability space of $\bs{r}_N$ and $\Sigma_N$.

%In this context, the equilibrium vector $\bs{x}_N^\star,$ and therefore the proportion 
%of non-zero components or the empirical 
%measure $\mu^{\bs{x}_N^\star}$ are also random  (defined on the same 
%probability space as $\bs{r}_N$ and $\Sigma_N$). 

%For a LV dynamical system, it is well-known that the property $\bs{x}N(0)\in (0,\infty)^N$ is maintained for all $t>0$ and $\bs{x}_N(t)\in (0,\infty)^N$. However, notice that in general, the
%equilibrium vector $\bs{x}_N^\star$ lies at the boundary of $\RR_+^N$, i.e. may have vanishing components. 

%in these
%conditions, it is of particular interest to evaluate the asymptotic behavior of
%$\| x^\star \|_0 / N$, which is the proportion of surviving species at the
%equilibrium (here $\| x \|_0=\mathrm{card}\left\{ i;\ x^\star_i>0\right\}$ is the number of non-zero elements of the vector
%$x$). 

Among the classical interaction matrix models considered in the literature
devoted to large LV systems are the Gaussian Orthogonal Ensemble (GOE) model,
the real Ginibre model (i.i.d.~centered Gaussian entries for $\Sigma_N$), or
the so-called elliptical model, that can be seen as an interpolation between
the GOE and the real Ginibre models \cite{all-tan-12}. For these models,
feasible equilibria where $x_i^\star>0$ for $1\le i\le N$ are
studied in \cite{bizeul2021positive,cle-fer-naj-22,akjouj2022feasibility,clenet2023equilibrium}. 

The large-$N$ properties of $\bs{x}_N^\star$ were
recently considered in the theoretical ecology literature.  In \cite{bun-17}, Bunin considered
a non-centered elliptical model with the help of the dynamical cavity method. A similar result
was obtained by Galla in \cite{gal-18} by means of generating functionals
techniques, see also \cite{opp-die-92,tok-04}.  Many insights are provided by these techniques from a physicist point of view.
However, up to our knowledge, no rigorous method to describe the asymptotic properties of $\bs{x}_N^\star$ can be found in the literature so far. 

The purpose of this paper is to address this question in the case where the interaction matrix is either taken from the GOE or follows a Wishart distribution. Our results on the asymptotics of $\mu^{\bs{x}_N^\star}$ mathematically confirm Bunin and Galla's works. 

\subsection*{Linear Complementarity Problem}
When it exists, the globally stable equilibrium
$\bs{x}_N^\star=(x_i^\star)$ of the LV equation above is known to be the solution of a so-called
Linear Complementarity Problem (LCP), see for instance \cite[Chap. 3]{tak-livre96}, which consists in
finding a vector with real entries that satisfies a system of inequalities involving
matrix $\Sigma_N$ and vector $\bs{r}_N$:
\begin{equation}\label{eq:LCP-first-sight}
\begin{cases}
x^\star_i &\ge 0,\\x^\star_i \left( r_i - \left[(I_N-\Sigma_N)\bs{x}^\star_N\right]_i\right)&=0,\\
r_i - \left[(I_N-\Sigma_N)\bs{x}^\star_N\right]_i &\le 0,
\end{cases}\qquad \textrm{for all}\ i\in \{ 1,\cdots, N\}\,.
\end{equation}
The two first conditions are natural for an equilibrium to system \eqref{eq:LV-system}. The third one is necessary for its stability and also admits an ecological interpretation related to the notion of non-invasibility. Sufficient conditions on $\Sigma_N$ to ensure existence and uniqueness of the solution $\bs{x}^\star_N$ are known. The problem boils down to the following question: how can we asymptotically extract statistical information on $\bs{x}^\star_N$, solution to the highly non-linear problem \eqref{eq:LCP-first-sight}, given that $\Sigma_N$ and $\bs{r}_N$ are random? 

The reader is referred to
Section~\ref{subsec-lcp} below for a quick overview of the LCP theory, and to
\cite{cot-pan-sto-livre09,murty1988linear} for complete and comprehensive expositions. 

%The principle of our approach stands as follows.

\subsection*{Approximate Message Passing}
The idea we develop in this paper is that the distribution
$\mu^{\bs{x}_N^\star}$ can be estimated for large $N$ by designing a proper Approximate
Message Passing (AMP) algorithm. 

Approximate Message Passing (AMP) is a technique that has
recently aroused an intense research effort in the fields of statistical
physics, machine learning, high-dimensional statistics and communication theory. 
Among the many landmark articles, we can cite
\cite{donoho2009message}, \cite{bay-mon-11}, \cite{bol-14}. More references can be found in the recent tutorial \cite{fen-etal-(now)22}. 

An AMP algorithm produces a sequence of $\RR^N$--valued
random vectors, say $\bs{\xi}^k=(\xi^k_i)$, which are iteratively built around a $N\times N$ random matrix, sometimes called the measurement matrix. This algorithm is conceived in such a way that for any finite collection $\bs{\xi}^1,\cdots, \bs{\xi}^k$ of these vectors, the following joint empirical 
distribution:
$$
\frac 1N \sum_{i=1}^N  \delta_{\xi^1_i,\cdots,\xi^k_i}
$$
converges as $N\to\infty$ to a Gaussian distribution on $\RR^k$ whose parameters can be fully characterized by the
so-called Density Evolution (DE) equations. In the context of our LV
equilibrium problem, it turns out that an AMP algorithm can be designed in such
a way that the AMP iterates approximate our LCP solution after an adequate
transformation. Thanks to this approximation, the asymptotic properties of 
$\mu^{\bs{x}_N^\star}$ can be deduced from the DE equations. 

\subsection*{Random matrix model and perspectives} Regarding the statistical model for $\Sigma_N$, we shall consider in this paper
the GOE model \cite{all-tan-12}, and the so-called Wishart model.  The latter
is a particular case of a kernel matrix, which is considered when the
interaction between two species depends on a distance between the values of
some functional traits attached to these species, see \cite[\S
4.6]{akj-etal-(arxiv)22} and the references therein, or the recent paper
\cite{roz-cru-gal-(arxiv)23}. Both models rely on Gaussian random variables, see Assumptions \ref{ass:goe}-\ref{ass:Wishart-matrix}, but we also provide results beyond the Gaussian case, see Assumptions \ref{ass:univ-goe}-\ref{ass:univ-wish}.

We believe that this LCP/AMP approach for studying $\mu^{\bs{x}^\star_N}$ can
be generalized and applied to more complex models for the interaction matrix
$\Sigma_N$. 
%We also advocate that the LCP/AMP approach for studying $\mu^{x^\star}$ can be
%generalized and applied to more involved models for the interaction matrix
%$\Sigma$. 
For instance, the recent results of Fan \cite{fan-22} might be used to cover
the general rotationally invariant case; more general models are also considered in
\cite{bay-lel-mon-15,wan-zho-fan-(arxiv)22}. Matrices with a variance profile,
possibly sparse \cite{hachem2023approximate}, and non-symmetric matrix models
(i.i.d.~elements, elliptical models) could also be considered as well. Some of
these generalizations are currently under investigation. 

\subsection*{Outline of the article} The problem statement, the main results and simulations are presented in Section \ref{sec:main}. In Section \ref{subsec:statement-GOE} (resp. Section \ref{subsec:statement-Wishart}) Theorem \ref{th:main-wigner} (resp. Theorem \ref{th:wishart-main}) describes the statistical properties of the equilibrium for an interaction matrix drawn from the GOE (resp. from the Wishart ensemble). In Section \ref{subsec:statement-universal}, we extend these results to matrix ensembles based on non-Gaussian entries. 
 Section \ref{sec:proof-GOE} is devoted to the proof of Theorem \ref{th:main-wigner}, starting with an outline of the proof in Section \ref{subsec:outline}, while elements of proof of Theorem \ref{th:wishart-main} are provided in Section \ref{sec:proof-wishart}. 

\subsection*{Main notations} For $x\in \RR$, let $x_+=\max(x,0)$, $x_-=\max(-x,0)$ and $[N]=\{ 1,\cdots, N\}$. Vectors will be denoted by lowercase bold letters $\bs{a}=(a_i),\bs{b}=(b_i)$, etc. If $f:\mathbb{R}\to \RR$ is a real function, vector $f(\bs{a})$ is defined pointwise by $f(\bs{a})= (f(a_i))_{i\in [N]}$. For vectors of same dimensions, $\bs{a}\odot \bs{b}= (a_ib_i)$ denotes the componentwise (Hadamard) product. 
Vector $\bs{1}_N$ is the $N\times 1$ vector of ones and $x\mapsto 1_{\mathcal S}(x)$ is the indicator function of set ${\mathcal S}$. Transpose of matrix $A$ is $A^\T$ and its eigenvalues are $\lambda_i(A)$.

For $\bs{a}=(a_i)$, $\bs{a}\succcurlyeq 0$ (resp. $\bs{a} \succ 0$) refers to the pointwise inequalities $a_i\ge 0$ (resp. $a_i>0$) for all $i\in [N]$. A positive (resp. negative) definite matrix $A$ is denoted by $A>0$ (resp. $A<0$).

Given a vector $\bs{a}$ and a matrix $A$, $\| \bs{a}\|$ denotes the Euclidian norm of $\bs{a}$ and $\|A\|$ the spectral norm of $A$. For a vector $\bs{a}$, 
$\| \bs{a}\|_0$ is the number of its non-zero elements and $\support(\bs{a})$ is its support, that is the set of indices of non-zero elements.

Given vectors $\bs{a}=(a_i)$, $\bs{a}^1=(a_i^1),\cdots, \bs{a}^k=(a_i^k)$ of 
the same size $N$, we denote as $\mu^{\bs{a}}$ and 
$\mu^{\bs{a}^1,\cdots, \bs{a}^k}$ the probability measures 
$$
\mu^{\bs{a}}=\frac 1N \sum_{i\in [N]} \delta_{a_i}\qquad \textrm{and}\qquad 
\mu^{\bs{a}^1,\cdots, \bs{a}^k} = \frac 1N \sum_{i\in [N]} \delta_{(a_i^1,\cdots, a_i^k)}\, .
$$
We call $\mu^{\bs{a}}$ the \emph{empirical distribution} of the  components of $\bs{a}$ and $\mu^{\bs{a}^1,\cdots, \bs{a}^k}$
the \emph{joint empirical distribution} of the components of $\bs{a}^1,\cdots, \bs{a}^k$.

If $\mu_N,\mu$ are probability measures over $\RR^d$ then $\mu_N\xrightarrow[N\to\infty]{w} \mu$ stands for the weak convergence of probability measures. The distribution of a random variable $X$ is denoted by ${\mathcal L}(X)$ and we express that two random variables $X,Y$ have the same distribution by $X\eqlaw Y$.
As usual, abbreviation a.s. stands for almost sure/surely.

\subsection*{Acknowledgment} These problems have been discussed at length with Matthieu Barbier, Maxime Clenet, François Massol and Chi Tran, whom we warmly thank. 

All authors are supported by the CNRS 80 prime project KARATE and GdR MEGA.
I.A. and M.M. are funded by Labex CEMPI (ANR-11-LABX-0007-01). 
W.H. and J.N. are supported by Labex B\'ezout (ANR-10-LABX-0058). 
 M.M. thanks for its hospitality the CRM Montr\'eal (IRL CNRS 3457) where part of this work was completed.\\

\section{Problem statement, assumptions, and main results}
\label{sec:main} 
\subsection{Equilibria, Wasserstein space and pseudo-Lipschitz functions} 

Independently of the struture of $\Sigma_N$, it is known that if $\| \Sigma_N
\| < 1$, then the ODE~\eqref{eq:LV-system} admits a unique solution
$(\bs{x}_N(t), t\ge 0)$ with a bounded trajectory, for any arbitrary initial
value $\bs{x}_N(0) \succ 0$, see \cite{li-etal-09}.  Moreover the same
condition $\| \Sigma_N \| < 1$ guarantees, as we shall recall in more detail in
Section \ref{sec:proof-GOE}, the existence of a globally stable equilibrium
point $\bs{x}^\star_N$ in the classical sense of the Lyapounov theory
\cite[Chapter~3]{tak-livre96}.  

Given  $k \in \mathbb N^*$, the \emph{Wasserstein space} ${\mathcal P}_k(\RR^d)$
is defined as the set of probability measures $\mu$ over $\RR^d$ with finite
$k^{\text{th}}$ moment: $\int_{\RR^d} \| \bs{x}\|^k \mu(d\bs{x}) <\infty$.
Given $\mu,\nu\in {\mathcal P}_k(\RR^d)$, we denote by ${\mathcal M}_k(\mu, \nu)$
the set of probability measures in ${\mathcal P}_k(\RR^d\times \RR^d)$ with
marginals $\mu$ and $\nu$, i.e. 
$$
\eta\in {\mathcal M}_k(\mu, \nu) \qquad \Rightarrow \qquad \begin{cases} \eta(A\times \RR^d) &= \mu(A)\,,\\
\eta(\RR^d\times B) &= \nu(B)\,,
\end{cases}
$$
for all $A,B$ Borel sets in $\mathbb{R}^d$. We can endow the space 
${\mathcal P}_k(\RR^d)$ with the distance:
$$
d_k(\mu,\nu) = \inf_{\eta\in {\mathcal M}_k(\mu,\nu)}
\left\{ \int_{\RR^d\times \RR^d} \| \bs{x} - \bs{y}\|^k \eta(d\bs{x}d\bs{y})
 \right\}^{1/k}\, .
$$
A function $\varphi:\RR^d\to \RR$ is \emph{pseudo-Lipschitz} with constant $L$ 
and degree $k\geq 2$ if for all $\bs{x},\bs{y}\in \RR^d$, the following 
inequality holds:
$$
|\varphi(\bs{x}) - \varphi(\bs{y}) | \quad \le \quad L \| \bs{x}- \bs{y}\| 
 \left( 1+\|\bs{x}\|^{k-1} +\|\bs{y}\|^{k-1}\right)\, .
$$
We denote by $PL_k(\RR^d)$ this set of functions. We will rely on the following classical lemma in the sequel, see for instance \cite[Section 1.1 and 7.4]{fen-etal-(now)22} and \cite{vil-livre09}.
\begin{lemma}\label{lemma:conv-P2}
Let $\mu_N,\mu\in {\mathcal P}_k(\RR^d)$ for $k\geq 2$. The following 
conditions are equivalent:
\begin{enumerate}[label=(\roman*)]
    \item $d_k(\mu_N,\mu) \xrightarrow[N\to \infty]{} 0$,
    \item For all $\varphi \in PL_k(\RR^d)$, 
  $\int \varphi d\mu_N \xrightarrow[N\to\infty]{} \int \varphi d\mu$,
    \item $\mu_N\xrightarrow[N\to\infty]{w} \mu$ and 
   $\int_{\RR^d} \| \bs{x}\|^k \mu_N(d\bs{x}) \xrightarrow[N\to\infty]{} 
    \int_{\RR^d} \| \bs{x}\|^k \mu(d\bs{x})$.
\end{enumerate}
\end{lemma}
If one of the equivalent conditions of Lemma \ref{lemma:conv-P2} is satisfied, we say that the sequence $(\mu_N)$ converges in 
${\mathcal P}_k(\RR^d)$ to $\mu$ and denote it by
$$
\mu_N \xrightarrow[N\to\infty]{{\mathcal P}_k(\RR^d)} \mu\, .
$$
If not misleading, we will occasionally drop $\RR^d$ and simply write 
${\mathcal P}_k, PL_k$.

Let $\bs{r}_N$ be a random vector of dimension $N\times 1$ that
satisfies the following assumption. 
\begin{assumption}
\label{ass:r}
The following hold true. 
\begin{enumerate}[label=(\roman*)]
    \item For all $N\ge 1$, $\bs{r}_N\succcurlyeq 0$ is defined on the same probability space as matrix $\Sigma_N$ and is independent of $\Sigma_N$.
    \item There exists a probability measure $\bar\mu \in \cP_2(\RR^+)$ such that  
$\bar\mu\neq \bs\delta_0$ and 
\[
(a.s.) \qquad \mu^{\bs{r}_N} \xrightarrow[N\to\infty]{\cP_2(\RR)} \bar\mu \ .
\]
\end{enumerate}

\end{assumption} 

\subsection{The GOE case} \label{subsec:statement-GOE} 
We first define rigorously the symmetric 
interaction matrix $\Sigma_N$ and express sufficient conditions for the existence of a unique global equilibrium $\bs{x}^\star_N$ to \eqref{eq:LV-system}.
\begin{assumption}\label{ass:goe}
    Let $A_N$ be a $N\times N$ matrix from the Gaussian Orthogonal Ensemble. Namely, considering that $X_N$ is a real $N\times
N$ matrix with independent $\cN(0,1)$ elements, 
$$A_N \eqlaw \frac{X_N + X_N^\T}{\sqrt{2}}\, .$$
Let $\kappa$ be a positive real number. Then, 
\begin{equation}\label{eq:GOE-matrix}
\Sigma_N = \frac{A_N}{\kappa \sqrt{N}}\, .
\end{equation}
\end{assumption}
\begin{remark} Denote by $A^{(N)}_{ij}$ the element $(i,j)$ of $A_N$, then
$A^{(N)}_{ij}=A^{(N)}_{ji}$ and ${\mathcal L} (A^{(N)}_{ij}) = \cN(0, 1+\delta_{ij})$ where $\delta_{ij}$ is the Kronecker symbol with value 1 if $i=j$, zero else. Much is known about this model, in particular the asymptotic behaviour of the spectral measure of $A_N/\sqrt{N}$ (Wigner's theorem) and its spectral norm, see for instance \cite{bai-sil-book,pas-livre} and the references therein:
\begin{equation}\label{eq:wigner-reminder}
(a.s.) \qquad \frac 1N \sum_{i\in [N]} \delta_{\lambda_i\left(A_N/\sqrt{N}\right)} \ \xrightarrow[N\to\infty]{w}\  \frac{\sqrt{(4-x^2)_+}}{2\pi}\, dx\qquad \textrm{and}\qquad \left\|\frac {A_N}{\sqrt{N}}\right\| \ \xrightarrow[N\to\infty]{}\  2\, .  
\end{equation}
\end{remark}

%For each $N$, consider the LV dynamical system with trajectories $x_N(t)\succcurlyeq 0$ represented on $t \in \RR_+$ by the Ordinary Differential Equation (ODE) 
%\begin{equation}
%\label{lv} 
%\dot x_N(t) = x_N(t) \odot 
%  \left( r_N + \left(\Sigma_N - I_N \right) x_N(t) \right), 
%  \quad 
% x_N(0) \in (0,\infty)^N ,  
%\end{equation} 
%where $u \odot v = (u_iv_i)_{1\le i\le N}$ for $u=(u_i)$ and $v=(v_i)$ $N\times 1$ vectors.
We shall consider the following assumption:  
\begin{assumption}
\label{ass:kappa}
The normalizing factor $\kappa$ in \eqref{eq:GOE-matrix} satisfies $\kappa > 2$.  
\end{assumption} 
Combining Assumption~\ref{ass:kappa} and the a.s. convergence of $\|A_N/\sqrt{N}\|$ toward 2, we get that with probability one, eventually
$$
\| \Sigma_N\|<1\, .
$$
Formally, this property means that there exists a set $\widetilde \Omega$ with probability one such that 
$$
\forall \omega\in \widetilde \Omega\,,\quad \exists N^\star(\omega)\,,\quad  \forall N \ge N^\star(\omega)\,, \qquad \| \Sigma_N\| <1\, .  
$$
As a consequence, for every $\omega\in \widetilde \Omega$, the existence and uniqueness of $\bs{x}^\star_N$ is granted for $N$ large enough.
% \wh{Je balance la condition $\kappa > 2$, loi du demi-cercle, etc. dès
% maintenant car il est important de garantir l'existence même de la solution
% avant de parler d'équilibre. Sinon le lecteur pourrait se demander si notre
% problème d'équilibre n'est pas vide ! \\ Il se trouve que $\kappa > 2$ est
% une condition trop forte pour l'existence mais bon pas besoin de s'embêter
% avec ça.  De toute façon on aura besoin du $\kappa > 2$ pour obtenir
% l'unicité de la solution du LCP et le caractère globalement stable de
% l'équilibre.} 
We can now describe the behaviour of the empirical distribution $\mu^{\bs{x}_N^\star}$ as $N\to \infty$
and state the main result of this section.
\begin{theorem}
\label{th:main-wigner} 
\begin{enumerate}[label=(\roman*)]
\item Let $\bar r \geq 0$ be
a real valued random variable with finite second moment and ${\mathcal L}(\bar{r})\neq \delta_0$. 
Let $\oZ$ be a 
$\cN(0, 1)$ random variable independent of $\bar r$. Then, for any $\kappa>\sqrt{2},$ the system of 
equations 
\begin{subequations}
\label{sys} 
\begin{align}
\kappa &= \delta + \frac{\gamma} \delta, \label{sys-kappa} \\ 
\sigma^2 &= \frac{1}{\delta^2} 
  \EE
   \left( \sigma \oZ + \bar r \right)_+^2, \label{sys-omega} \\ 
\gamma &= \mathbb{P} \Bigl[ \sigma \oZ + \bar r > 0 \Bigr]  , 
 \label{sys-gamma} 
\end{align}
\end{subequations} 
admits an unique solution $(\delta,\sigma,\gamma)$ in
$(1/\sqrt{2},\infty) \times (0,\infty) \times (0,1)$. 
\item Let $\bs{r}_N\succcurlyeq 0$ and let Assumptions \ref{ass:goe} and \ref{ass:kappa} hold true. Then, $\| \Sigma_N  \| < 1$ eventually with 
probability one. For such $N$'s, the 
ODE~\eqref{eq:LV-system} is defined for all $t\in \RR_+$ and has a globally stable 
equilibrium $\bs{x}^\star_N$. For the other $N$'s, let $\bs{x}^\star_N = 0$. 

\item Let Assumptions~\ref{ass:r}, \ref{ass:goe} and~\ref{ass:kappa} hold. Define $\bs{x}^\star_N$ as previously. The distribution $\mu^{\bs{x}^\star_N}$ is a $\cP_2(\RR)$--valued random variable 
on the probability space where $A_N$ and $\bs{r}_N$ are defined. Assume that $\bar{r}$ is a r.v. with ${\mathcal L}(\bar{r}) = \bar{\mu}$, independent of $\bar Z\sim {\mathcal N}(0,1)$. Then, the convergence 
\begin{equation}
\label{cvg-muN} 
(a.s.)\qquad \mu^{\bs{x}^\star_N} \ \xrightarrow[N\to\infty]{\cP_2(\RR)} \
 \mcL\left( \left( 1 + \gamma/\delta^2 \right) 
  \left( \sigma \oZ + \bar r \right)_+ \right) 
\end{equation} 
holds true, where $\delta, \sigma, \gamma$ are defined as solutions of system \eqref{sys}.
\end{enumerate}
\end{theorem}
This theorem, which proof is postponed to Section \ref{sec:proof-GOE}, calls for some remarks.  
\begin{remark}
Equations \eqref{sys-kappa}-\eqref{sys-gamma} have already been obtained\footnote{Notice that in \cite{bun-17,gal-18}, the authors consider more general models such as the elliptical model, which in particular captures the Wigner model.} at a physical level of rigor by Bunin \cite{bun-17} and Galla \cite{gal-18}. Up to our knowledge, Theorem \ref{th:main-wigner} is the first rigorous statement to describe the asymptotic properties of $\bs{x}^\star_N$. 
\end{remark}

\begin{remark}
Notice that system \eqref{sys} admits an unique solution for $\kappa >\sqrt{2}$
while Convergence \eqref{cvg-muN} is only established for $\kappa>2$.
%\walid{Jamal tu as évoqué le GRETSI je crois pour commenter cette remarque. 
%Le GRETSI montre que pour $\kappa > 2$ le semiflot existe et est borné. 
%Mais rien sur le comportement asymptotique. Donc je ne sais pas ce que le 
%GRETSI apporte finalement.} 
\end{remark}

\begin{remark}[behavior of surviving species proportion] 
Theorem \ref{th:main-wigner} sheds some light on the proportion of surviving species at equilibrium : inspecting~\eqref{sys-gamma} 
and~\eqref{cvg-muN}, the parameter 
$\gamma$ can be interpreted as the limiting proportion of surviving
species $\| x^\star_N \|_0 / N.$  Simulations in Fig. \ref{subfig:prop-GOE} confirm this fact.

One can see from Equation~\eqref{sys-gamma} 
 that $\gamma > 1/2,$ which means that in this model, more than half the species survive. 

Furthermore, an easy calculation involving Equations~\eqref{sys-omega}
and~\eqref{sys-gamma} shows that $\gamma$ does not change if we replace $\bar
r$ with $K \bar r$ where $K > 0$ is an arbitrary constant.

Nevertheless, on a rigorous level, one can only deduce from Theorem \ref{th:main-wigner} that
\[
\sup_\varphi \left\{  (a.s) \lim_{N\to\infty}
\frac 1N \sum_{i\in [N]} \varphi(x^\star_{i}) \right\}= \gamma,  
\]
where $\sup_\varphi$ is taken on the set of functions 
$\{ \varphi : \RR \to [0,1] \ \text{continuous}, \ \varphi(0) = 0 \}$. 

Since the function $1_{\{x > 0\}}$ is not continuous, the 
convergence~\eqref{cvg-muN} does not imply that $\| \bs{x}^\star_N \|_0 / N$ 
converges to $\gamma$, for any type of convergence.  
Up to our knowledge, the study of the asymptotic behavior of $\| \bs{x}^\star_N \|_0 / N$ is an open question.  
% \wh{Vaut mieux être
% prudent et ne rien dire sur les simulations à ce sujet (ou alors, en marchant
% sur des \oe ufs) car on ne sait pas trop ce que veut dire un zéro dans un
% ordinateur.} \\
% \wh{Existe-il des trucs sur le comportement des supports des solutions LCP 
% sous perturbation ? Peut-être.} 
\end{remark}

 \begin{figure}[h!]
  \centering
    \begin{subfigure}[b]{.46\linewidth}
        \includegraphics[scale=0.35]{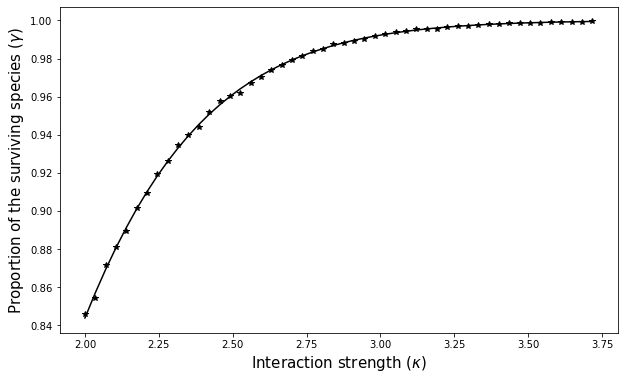}
        \caption{Proportion of surviving species (GOE)}
        \label{subfig:prop-GOE}
    \end{subfigure}%
    \hspace*{\fill} 
    \begin{subfigure}[b]{.46\linewidth}
        \includegraphics[scale=0.515]{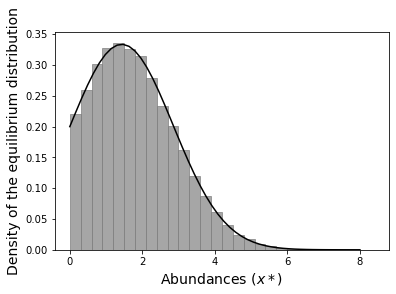}
        \caption{Density of a surviving species}
        \label{subfig:hist-GOE}
    \end{subfigure}
    
 \caption{Subplot \ref{subfig:prop-GOE} represents the proportion of surviving species, that is the proportion of positive components of the equilibrium $\bs{x}^\star$ (star), versus the theoretical value of $\gamma$ (solid line), given the parameter $\kappa$ which varies from 2 to 3.75. In the plot, $N=1000$ and each point (star) is the mean of proportions obtained out of 100 Monte-Carlo simulations. Subplot \ref{subfig:hist-GOE} represents the distribution of a surviving species ($N=1000$ and 100 Monte-Carlo simulations). The solid line represents the theoretical value of the density $f_{Z\mid Z>0}$ where $Z$ is the random variable with limiting distribution of $\mu^{\bs{x}_N^\star}$ given in \eqref{cvg-muN} - cf. Theorem \ref{th:main-wigner}.}
\label{fig:GOE}
\end{figure}

\subsection{The Wishart case} 
\label{subsec:statement-Wishart}
Wishart matrices are interesting in theoretical ecology to model interactions between two species which depend on the distance between values of some given functional traits, see for instance \cite[\S~4.6]{akj-etal-(arxiv)22} or \cite{rozas2023competitive}.  

    \begin{assumption}\label{ass:Wishart-matrix}
        Let $B_N$ be a $P\times N$ matrix with i.i.d. Gaussian ${\mathcal N}(0,1)$ entries. Let $\kappa$ be a real positive number and define the $N\times N$ matrix $\Sigma_N$ as:
        \begin{equation}\label{eq:Wishart-matrix}
            \Sigma_N = \frac{B_N^\T B_N}{\kappa P}\, .
        \end{equation}
    \end{assumption}
    For this model, the $i$th column of matrix $B_N$ is a vector modelling the traits of species $i$. 

    We will be interested in the specific regime where $N,P$ go to infinity at the same pace:
    \begin{assumption}\label{ass:MP}
        Let $N=N(P)$  and assume that $$
        \frac NP \xrightarrow[P\to\infty]{} c\in (0,\infty)\, .
        $$ 
        This regime will be denoted by $N,P\to \infty$ in the sequel.
    \end{assumption}

Model \eqref{eq:Wishart-matrix} has been thoroughly studied under Assumption \ref{ass:MP}. Marchenko-Pastur's theorem describes the asymptotic behaviour of the spectral limit of $B_N^\T B_N/P$. The limiting spectral norm has been studied by Bai and Yin, see for instance \cite{bai-sil-book,pas-livre} and the references therein:
$$
(a.s.) \qquad \left\|\frac {B_N^\T B_N}{P}\right\| \ \xrightarrow[N,P\to\infty]{}\  (1+\sqrt{c})^2\, .  
$$

	\begin{assumption}
		\label{ass:kappa-wishart} The normalizing factor in \eqref{eq:Wishart-matrix} satisfies $\kappa>(1+\sqrt{c})^2$.  
	\end{assumption} 
		
	We can now state the main result of this section.
	
	\begin{theorem}
		\label{th:wishart-main} 
        \begin{enumerate}[label=(\roman*)]
        \item Let $\bar r \geq 0$ be a real valued r.v. with ${\mathcal L}(\bar{r}) \neq \delta_0$. Let $\bar{Z}$ be a $\mathcal{N}(0, 1)$ r.v. independent of $\bar r$. Then, for every $\kappa>\left( 1+\sqrt{\frac c2}\right)^2$, the system of 
		equations 
		\begin{subequations}
			\label{eq:sys-wishart} 
			\begin{align}
			\kappa &= \left(\delta + c\gamma\right)\left(1+\frac{1}{\delta}\right), \label{sys-wishart-kappa} \\ 
			\tau^2 &= \frac{c}{\delta^2} 
			\mathbb{E} \left[
			\left( \tau \bar{Z} + \bar r \right)_+^2\right], \label{sys-wishart-omega} \\ 
			\gamma &= \mathbb{P} \Bigl[ \tau \bar{Z} + \bar r > 0 \Bigr]  , 
			\label{sys-wishart-gamma} 
			\end{align}
		\end{subequations} 
		admits an unique solution $(\delta,\tau,\gamma)$ in 
		$(\sqrt{c/2},\infty) \times (0,\infty) \times (0,1)$. 
		\item Let $\bs{r}_N\succcurlyeq 0$ and let Assumptions~\ref{ass:Wishart-matrix}, \ref{ass:MP} and \ref{ass:kappa-wishart} hold. 
		Then, $\| \Sigma_N \| < 1$ eventually with probability one. For such $N$'s, the LV ODE solution is defined for all $t\in \mathbb{R}_+$ and has a globally stable equilibrium $\boldsymbol{x}^\star_N$. For the other $N$, set $\boldsymbol{x}^\star_N = 0$. 
		
		\item Let Assumptions~\ref{ass:r}, \ref{ass:Wishart-matrix}, \ref{ass:MP} and \ref{ass:kappa-wishart} hold. Define $\bs{x}^\star_N$ as previously. The distribution $\mu^{\boldsymbol{x}^\star_N}$ is a $\mathcal{P}_2(\mathbb{R})$--valued random variable on the probability space where $A_N$ and $\bs{r}_N$ are defined. Assume that $\bar{r}$ is a r.v. with ${\mathcal L}(\bar{r})=\bar{\mu}$, independent of $\bar Z\sim {\mathcal N}(0,1)$. The following convergence holds true: 
		\begin{equation}
		\label{cvg-wishart-mu} 
		(a.s.)\qquad \mu^{\bs{x}^\star_N} \xrightarrow[N,P\to\infty]{\mathcal{P}_2(\mathbb{R})} 
		\mathcal L\left( \left( 1 + 1/\delta \right) 
		\left( \tau \bar{Z} + \bar r \right)_+ \right)\, ,
				\end{equation} 
    where $\delta$, $\tau$ and $\gamma$ are defined as solutions of system \eqref{eq:sys-wishart}.
    \end{enumerate}
	\end{theorem}
     There is a strong matching between the parameters obtained by solving system \eqref{eq:sys-wishart} and their empirical counterparts obtained by Monte-Carlo simulations, as illustrated in Fig. \ref{fig:wishart}.

 \begin{figure}[h!]
  \centering
    \begin{subfigure}{.46\linewidth}
        \includegraphics[scale=0.35]{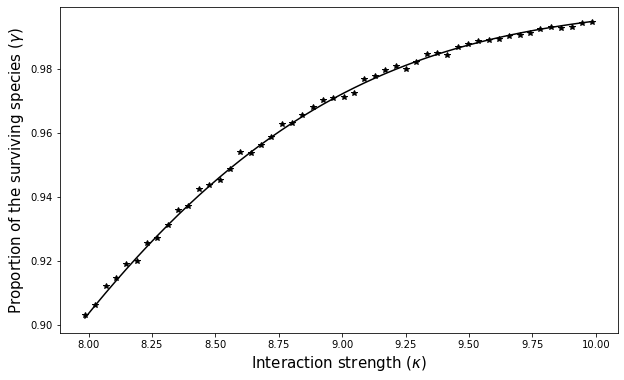}
        \caption{Proportion of surviving species (GOE)}
        \label{subfig:prop-Wishart}
    \end{subfigure}%
    \hspace*{\fill} 
    \begin{subfigure}{.46\linewidth}
        \includegraphics[scale=0.515]{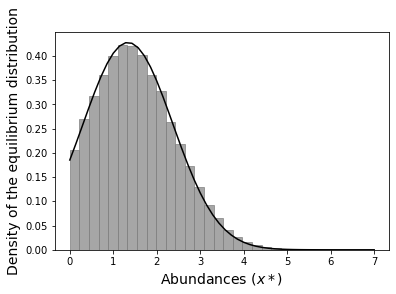}
        \caption{Density of a surviving species}
        \label{subfig:hist-Wishart}
    \end{subfigure}
    
 \caption{Subplot \ref{subfig:prop-Wishart} represents the proportion of surviving species, that is the proportion of positive components of the equilibrium $\bs{x}^\star$ (star), versus the theoretical value of $\gamma$ (solid line), given the parameter $\kappa$ which varies from 2 to 3.75. In the plot, $N=1000$, $P=300$ and each point (star) is the mean of proportions obtained out of 100 Monte-Carlo simulations. Subplot \ref{subfig:hist-Wishart} represents the distribution of a surviving species ($N=1000$, $P=300$ and 100 Monte-Carlo simulations). The solid line represents the theoretical value of the density $f_{Z\mid Z>0}$ where $Z$ is the random variable with limiting distribution of $\mu^{\bs{x}_N^\star}$ given in \eqref{cvg-wishart-mu} - cf. Theorem \ref{th:wishart-main}.}
\label{fig:wishart}
\end{figure}

%  \begin{figure}[h!]
%  \centering
%    \begin{subfigure}{.46\linewidth}\label{Wishart_gamma}
%        \includegraphics[scale=0.35]{Wishart_gamma.png}
%        \caption{essai}
%    \end{subfigure}%
%    \newline
   % \hspace*{\fill}   % maximize separation between the subfigures
%    \begin{subfigure}{width=\textwidth}\label{Wishart_mean}
%        \includegraphics[scale=0.35]{Wishart_mean.png}
%    \end{subfigure}%
%    \hspace*{\fill}  
%    \begin{subfigure}{width=\textwidth}\label{Wishart_moment2}
%        \includegraphics[scale=0.35]{Wishart_moment2.png}
%    \end{subfigure}%
% \caption{The plots represent a comparison between the theoretical solutions $\gamma$, the mean and the root mean square at equilibrium of the abundances arising from system \ref{sys-wishart} and their empirical Monte Carlo counterpart (the star marker) as functions of the interaction strength $\kappa$. Matrix $\Sigma_n$ has size $n = 500$ and $p=300$. The number of Monte Carlo experiments is $50$. When interaction $\kappa^{-1}$ increases, the proportion of surviving species $\gamma$ decrease but their mean and variance increase.  }   \label{fig:Wishart}
%\end{figure}
 
	The proof of this theorem relies on an asymmetric version of the AMP 
	algorithm and is otherwise very close to the proof of Theorem \ref{th:main-wigner}. We provide some details in Section \ref{sec:proof-wishart}.  
	
\subsection{Towards universality}
\label{subsec:statement-universal}

We mentionned in the introduction that AMP techniques have been generalized to matrices with non-necessarily Gaussian entries, see \cite{bay-lel-mon-15,che-lam-21,dud-lu-sen-(arxiv)22,wan-zho-fan-(arxiv)22}. It is possible, at low cost, to relax the Gaussiannity assumption of the entries in Assumptions \ref{ass:goe} and \ref{ass:MP}. 

We first strenghten Assumption \ref{ass:r} and replace it by the following stronger assumption: 

\begin{assumption}
    \label{ass:r-strong}
The following holds true: 
\begin{enumerate}[label=(\roman*)]
    \item For all $N\ge 1$, $\bs{r}_N\succcurlyeq 0$ is defined on the same space as 
  matrix $\Sigma_N$ and is independent of $\Sigma_N$.
  \item There exists a probability measure 
  $\bar\mu \in \cP(\RR^+)$ such that $\bar\mu\neq \bs\delta_0$, the moment
generating function of $\bar\mu$ is analytical near zero (which implies 
that $\bar\mu$ has all its moments finite), and 
\[
(a.s.) \qquad \mu^{\bs{r}_N} \xrightarrow[N\to\infty]{\cP_k(\RR)} \bar\mu \qquad \textrm{for all}\quad k\ge 1\,.
\]
\end{enumerate}
\end{assumption} 
We now relax the GOE assumption (Assumption \ref{ass:goe}).
\begin{assumption}
\label{ass:univ-goe} 
Let $A_N= \left( A_{ij}^{(N)}\right)$ be a $N\times N$ symmetric matrix where
the $A^{(N)}_{ij}$'s are centered independent random variables satisfying 
$$
\EE (A^{(N)}_{ij})^2 = 1\quad (i<j)\,,\quad \sup_N \max_i \EE (A^{(N)}_{ii})^2 < C\,,
$$
and 
$$
\max_{i,j} 
  N^{1-k/2} \EE \left| A^{(N)}_{ij} \right|^k \xrightarrow[N\to\infty]{} 0\quad (k\ge 3)\,.
$$
Moreover, the following holds true:
\begin{equation}
\label{cvg-normsym} 
\| A_N \| \toaslong 2 .
\end{equation} 
Denote by $\Sigma_N = A_N / (\kappa\sqrt{N})$.
\end{assumption} 

\begin{remark}[Wigner matrices]\label{rem:wigner}
The standard example of a matrix $A_N$ that generalizes the GOE model and that
complies with Assumption~\ref{ass:univ-goe} corresponds to the case where
$A^{(N)}_{ij} \eqlaw \chi$ for $i\neq j$ and   $A^{(N)}_{ii} \eqlaw \chi'$,
where the centered random variables $\chi$ and $\chi'$ do not depend on $N$,
$\EE \chi^2 = 1$, and $\chi$ and $\chi'$ have all their moments finite.  Note
that in this case, the convergence~\eqref{cvg-normsym} is a standard result in
Random Matrix theory \cite{bai-sil-book,pas-livre}. 
\end{remark}

\begin{remark}[Sparse models]
Sparsity of the food interactions is often justified from an ecological point of view, see~\cite{bus-etal-17}. 

Beyond the model described in Remark \ref{rem:wigner}, some
sparse models can also be covered by Assumption~\ref{ass:univ-goe}, as the
following example shows: Let $p_N \in (0,1),$ and
\[
A^{(N)}_{ij} = \left\{\begin{array}{cl}  
 1 / \sqrt{p_N} & \text{with probability } p_N / 2 \\
-1 / \sqrt{p_N} & \text{with probability } p_N /2 \\ 
 0 & \text{with probability } 1 - p_N . 
 \end{array}\right. 
\]
Since $\EE \bigl| A^{(N)}_{ij} \bigr|^k = p_N^{1-k/2}$, the moment condition in
Assumption~\ref{ass:univ-goe} is satisfied as soon as $N p_N \xrightarrow[N\rightarrow\infty]{} \infty$.
Furthermore, the spectral norm convergence condition~\eqref{cvg-normsym} is
satisfied when $\frac{Np_N}{\log N} \xrightarrow[N\rightarrow\infty]{} \infty$, as shown in \cite{ben-bor-kno-20}, see also
\cite{ben-bor-kno-19}. Therefore, according to this model, a species within our
LV system can interact with an average number of species much smaller than $N$ but of an order 
$\gg \log N$. 
\end{remark}

We are now in position to state a non-Gaussian version of Theorem \ref{th:main-wigner}:

\begin{theorem}[Non-Gaussian symmetric matrix]
\label{th:univ-goe} 
All the conclusions of Theorem~\ref{th:main-wigner} remain true if 
Assumptions~\ref{ass:r} and~\ref{ass:goe} in the statement of this theorem are
replaced with Assumptions~\ref{ass:r-strong} and~\ref{ass:univ-goe} 
respectively.
\end{theorem} 

Elements of proof are provided in Appendix \ref{app:universality}.

We now provide the proper assumption to state a non-Gaussian version of Theorem \ref{th:wishart-main}.

\begin{assumption}
\label{ass:univ-wish} 
\begin{itemize}
\item The $P\times N$ random matrix $B_N = \left( B^{(N)}_{ij} \right)_{i,j=1}^{P,N}$
is such that the random variables $B^{(N)}_{ij}$ for $i\in[P]$ and $j\in[N]$ 
are centered, independent, with variance one and satisfy
\[
\max_{i,j} 
  P^{1-k/2} \EE \bigl| B^{(N)}_{ij} \bigr|^k \xrightarrow[N\to\infty]{} 0\,,\qquad (k\ge 3)\, .
\]
We denote by 
$$
\Sigma_N = \frac{B_N^\T B_N}{\kappa P}\, .
$$
\item Moreover, $N = N(P)$, and there exists $c > 0$ such that 
\[
\frac{N(P)}{N} \xrightarrow[P\to\infty]{} c . 
\]
\item Finally, in this asymptotic regime, the convergence 
\begin{equation} 
\label{cvg-normasym}
\left\|\frac {B_N^\T B_N}{P}\right\| \ \xrightarrow[P\to\infty]{\text{a.s.}}
  \  (1+\sqrt{c})^2  
\end{equation} 
holds true. 
\end{itemize}
\end{assumption} 
\begin{remark}
The standard model for a matrix $B_N$ satisfying this assumption is the model
for which $B^{(N)}_{ij} \eqlaw \chi$, where $\chi$ is a centered
random variable with unit variance having all its moments finite. In this case, the 
convergence~\eqref{cvg-normasym} is a standard random matrix theory result 
\cite{bai-sil-book,pas-livre}. 
\end{remark}

With this assumption at hand, we are in position to provide a counterpart to Theorem \ref{th:wishart-main}.

\begin{theorem}[Non-Gaussian Wishart matrices]
\label{th:asym-univ} 
All the conclusions of Theorem~\ref{th:wishart-main} remain true if
Assumption~\ref{ass:r} is replaced with Assumption~\ref{ass:r-strong} and 
Assumptions~\ref{ass:Wishart-matrix} and~\ref{ass:MP} are replaced with Assumption~\ref{ass:univ-wish} in the statement of this 
theorem.
\end{theorem} 
Elements of proof are provided in Appendix \ref{app:universality}.

\section{Proof of Theorem~\ref{th:main-wigner}} 
\label{sec:proof-GOE} 

\subsection{Outline of the proof} \label{subsec:outline}
There are four steps in the proof. 

\subsubsection*{Step 1} In Section \ref{subsec:system}, we establish the uniqueness and existence of parameters $\delta$, $\sigma$ and $\gamma$, solutions to system \eqref{sys}. These parameters will play a crucial role to design an AMP algorithm fitted for our purpose. Equations \eqref{sys-kappa}-\eqref{sys-gamma} will progressively appear during the proof.

\subsubsection*{Step 2} In Section \ref{subsec-lcp}, we characterize the stable equilibrium $\bs{x}^\star_N$ of \eqref{eq:LV-system} as the solution of a Linear Complementarity Problem (LCP). We give an equivalent formulation of the solution of a LCP as the solution of a fixed-point equation, see Proposition
\ref{prop:fixed-point-LCP}.

\subsubsection*{Step 3} In Section \ref{subsec:GOE-algo}, we first recall some general facts about Approximate Message Passing (AMP) algorithms and present a specific algorithm \eqref{amp} whose output $(\bs{\xi}^k_N)_+$  will converge toward $\bs{x}^\star_N$, characterized as the solution of the fixed-point equation associated to the corresponding LCP. The approximate fixed-point equation satisfied by $\bs{\xi}^k_N$ is given in \eqref{eq:LCP-almost}, see also \eqref{eq:LCP-almost-bis}.

\subsubsection*{Step 4} The strength of the AMP procedure is that we can track down via the Density Evolution (DE) equations the asymptotic distribution of $(\bs{\xi}_N^k)_+$'s empirical measure for any $k$. We can then transfer if to $\bs{x}^\star_N$ by using a perturbation result by Chen and Xiang in \cite{che-xia-07}, see \eqref{eq:bound-LCP}.
A central argument borrowed from Montanari and Richard \cite{mon-ric-16} is that vectors $\bs{\xi}^k_N$ tend to be aligned for large $k$.

\subsection{Existence and uniqueness of the solution of system~\eqref{sys}} 
\label{subsec:system}

We begin with the following technical lemma, the third part of which will be used in Section \ref{subsec:GOE-algo}. To avoid any ambiguity, we shall always refer to $\sigma$ as the unique positive root of $\sigma^2>0$. 
\begin{lemma}
\label{lemma:point-fixe}
Let $\bar r$ be a non negative r.v. with ${\mathcal L}(\bar r)\neq \delta_0$.% and $\kappa>\sqrt{2}$\footnote{\jamal{On n'a pas besoin de la condition $\kappa>\sqrt{2}$ dans ce lemme. $\kappa$ n'intervient pas. }}. 
\begin{enumerate}[label=(\roman*)]
\item For a given $\delta > 0$, Equation~\eqref{sys-omega} admits a solution
$\sigma^2$ if and only if $\delta > 1 / \sqrt{2}$. In this case, this 
solution is unique, and is denoted by $\sigma^2(\delta)$. 
\item Let $\delta > 1 / \sqrt{2}$ then
$$
\mathbb{P} \{ \sigma(\delta) \oZ +\bar r \ge 0 \} < \delta^2\, .
$$
\item Assume $\delta > 1 / \sqrt{2}$. Starting with an arbitrary $\sigma_0 \geq 0$, 
consider the iterative scheme:
$$
\sigma^2_{t+1} = \frac{1}{\delta^2} 
  \EE \left( \sigma_t \oZ + \bar r \right)_+^2\,, \qquad \textrm{then}\quad \sigma^2_{t} \xrightarrow[t\to\infty]{} \sigma^2(\delta)\, .$$
\end{enumerate}
\end{lemma}
Proof of Lemma \ref{lemma:point-fixe} is postponed to Appendix \ref{proof:point-fixe}.

\vspace{0.25cm}

We now establish that system~\eqref{sys} has a unique solution
$$(\delta,\sigma,\gamma) \in (1/\sqrt{2},\infty) \times (0,\infty) \times
(0,1)\,.$$ Let $\delta>1/\sqrt{2}$, $\sigma^2(\delta)$ be defined by \eqref{sys-omega}, and $\gamma(\delta)$ by \eqref{sys-gamma}. Setting $f(\sigma^2) = \mathbb{E} (\sigma \oZ +\bar r)_+^2$, we have established in the proof of Lemma \ref{lemma:point-fixe}-(i) that 
$$
\gamma(\delta) = \frac{df}{d\sigma^2}\bigg|_{\sigma^2=\sigma^2(\delta)}\,.
$$
Moreover $\gamma(\delta)<\delta^2$ by Lemma \ref{lemma:point-fixe}-(ii). All what remains to show is that the equation 
\begin{equation}
\label{kd} 
\kappa = \delta + \frac{\gamma(\delta)}{\delta} 
\end{equation} 
has a unique solution $\delta > 1 / \sqrt{2}$. We thus need to study the 
behavior of $\gamma(\delta)$. In all the remainder, differentiability issues 
can be easily checked and are skipped. 

Recall that $d f(\sigma^2)/d\sigma^2$ decreases asymptotically to
$1/2$ as $\sigma^2$ increases from $0$ to $\infty$, from which we can deduce that $\sigma^2(\delta) \to \infty$ as 
$\delta \downarrow 1 / \sqrt{2}$ by Lemma \ref{lemma:point-fixe}-(ii). Using the fact that 
$$
\sigma^2(\delta) = \frac{f(\sigma^2(\delta))}{\delta^2}$$
and taking the derivatives with respect to $\delta$, we get that 
\[
\frac{d\sigma^2(\delta)}{d\delta} \Bigl( 
 1 - \frac{1}{\delta^2} 
 \left.\frac{d f(\sigma^2)}{d\sigma^2}\right|_{\sigma^2 = \sigma^2(\delta)} 
 \Bigr) = - \frac{2 f(\sigma^2(\delta))}{\delta^3}, 
\]
which shows that $\sigma^2(\delta)$ is a decreasing function. Hence $\gamma(\delta)$ is increasing since $\sigma \mapsto \mathbb{P}\{\sigma \oZ +\bar r\ge 0)$ is decreasing (cf. proof of Lemma \ref{lemma:point-fixe}).

We can now conclude. For $\delta \downarrow 1/\sqrt{2}$, $\sigma^2(\delta) \to \infty$ by what
precedes, thus, $\gamma(\delta) \downarrow 1/2$, and   
$\delta + \gamma(\delta) / \delta \to \sqrt{2} < \kappa$. 
Near infinity, $\delta + \gamma(\delta) / \delta \sim 
\delta > \kappa$. Consequently, Eq.~\eqref{kd} has a solution by 
continuity. To establish uniqueness, we prove that the function $\delta\mapsto 
\delta + \gamma(\delta) / \delta$ is increasing. Indeed, 
\[
\frac{d}{d\delta}\left( \delta + \frac{\gamma(\delta)}{\delta} \right)
 = 1 + \frac{\gamma'(\delta)}{\delta} - \frac{\gamma(\delta)}{\delta^2} 
 \geq  1 - \frac{\gamma(\delta)}{\delta^2} > 0 
\]
as shown by Lemma \ref{lemma:point-fixe}-(ii), and we are done. Proof of Theorem \ref{th:main-wigner}-(i) is completed.

\subsection{Characterization of $x^\star_N$ through a LCP} 
\label{subsec-lcp} 

In this section, we recall the connection between the possible stable equilibrium of the ODE \eqref{eq:LV-system} and the solution of an underlying LCP in the theory of mathematical programming. We mainly rely on chapter 3 of Takeuchi's book \cite{tak-livre96}. 

% It is well-known that the stable equilibria of the ODE~\eqref{lv} can be seen as solutions of LCP's in the theory of mathematical programming \cite{tak-livre96}.  We recall herein some basics of the LCP theory in order to establish the first part of Theorem~\ref{main}. All the results in this subsection are known. 

Given a matrix $M \in \RR^{N\times N}$ and a vector $\bs{c} \in \RR^N$, the LCP
problem, denoted as $\LCP(M,\bs{c})$, consists in finding couples of vectors 
$(\bs{y}, \bs{w}) \in \RR^N \times \RR^N$ satisfying
\begin{equation}\label{eq:def-LCP}
\left\{
\begin{array}{lcl}
 \bs{w} &=& M \bs{y} + \bs{c} \succcurlyeq 0\,, \\ 
 \bs{y} &\succcurlyeq& 0\,, \\ 
 \bs{w}^\T \bs{y} &=&0\,.
 \end{array}
 \right.
\end{equation}
Notice that the last condition can be written equivalently either $w_i y_i=0$ for all $i\in [N]$ or 
$\support(\bs{w}) \cap \support(\bs{y}) = \emptyset$. When a solution $(\bs{y}, \bs{w})$ exists we write $\bs{y} \in \LCP(M,\bs{c})$. If a solution exists and is unique, we write
$$\bs{y} = \LCP(M,\bs{c})\,.$$
A necessary and sufficient condition for the existence of a unique solution to the LCP problem has been given by Murty \cite{murty1972number}, see also \cite{cot-pan-sto-livre09}. For a symmetric matrix, this condition is simply to be positive definite. 

The following proposition establishes a connection between the solution of an LCP problem and globally stable equilibrium for a LV system .  

\begin{proposition}[Lemma 3.2.2 and Theorem 3.2.1 of \cite{tak-livre96}] \label{prop:lyap}
Given a symmetric matrix $B\in\RR^{N\times N}$ and a vector $\bs{c}\in\RR^N$, 
consider the following LV system of ODE: 
\begin{equation}
\label{tak} 
 \frac{d\bs{y}}{dt} (t) = \bs{y}(t) \odot \left( \bs{c} + B \bs{y}(t) \right) \,,  \quad
 \bs{y}(0) \succ 0\, .
\end{equation} 
for all $t\ge 0$. Then, the LCP problem $\LCP(-B,-\bs{c})$ has an unique solution for each $\bs{c} \in
\RR^N$ if and only if $B<0$, i.e. $B$ is negative definite.
On the 
domain where $B < 0, c \in \RR^N$, the function $x = \LCP(-B,-c)$ is 
measurable. Moreover, if $B < 0$, then for every $\bs{c} \in \RR^N$, the ODE~\eqref{tak} has a 
globally stable equilibrium $\bs{y}^\star$ given by $\bs{y}^\star = \LCP(-B,-\bs{c})$. 
\end{proposition}

Indeed, the equilibrium is characterized by the conditions $\bs{y}^\star  \succcurlyeq 0$ and for all $i \in [N],
y_i^\star(c+(By^\star)_i)=0$ whereas the condition $-c - By^\star  \preccurlyeq 0$ (with the obvious meaning of $\preccurlyeq$) turns out to be a necessary
condition for the equilibrium $y^\star$ to be stable in the classical
sense of Lyapounov theory (see \cite[Chapter 3]{tak-livre96} to recall the
different notions of stability, and \cite[Theorem~3.2.5]{tak-livre96} for this
result).

Going back to system \eqref{eq:LV-system}, a potential equilibrium $\bs{x}_N^\star$ should satisfy
$$
\bs{x}_N^\star \succcurlyeq 0 \qquad \textrm{and} \qquad x^\star_i \left( 
r_i - \left[(I_N - \Sigma_N)\bs{x}^\star_N\right]_i \right)=0\quad \textrm{for all}\quad i\in [N]\, 
$$
and
\[
\bs{r}_N + \left(\Sigma_N - I_N \right) \bs{x}_N^\star  \preccurlyeq 0,
\]
which means that the couple $(\bs{x}_N^\star, \bs{w}_N^\star)$ solves the problem $\LCP(I_N - \Sigma_N, -\bs{r}_N)$.

Applying the reminder \eqref{eq:wigner-reminder} and Assumption \ref{ass:kappa}, matrix $I_N- \Sigma_N$ is eventually positive definite with probability one. Define now the vector $\bs{x}_N^\star$ by  
\begin{equation} 
\label{x*lcp} 
\bs{x}_N^\star = \left\{\begin{array}{ll}
\LCP(I_N - \Sigma_N, -\bs{r}_N) & \text{if } \| \Sigma_N \| < 1 , \\
 0 & \text{otherwise} \ .
 \end{array}\right. 
\end{equation}
Then, from Proposition \ref{prop:lyap}, we get that vector $\bs{x}_N^\star$ satisfies the statement of Theorem \ref{th:main-wigner}-(ii).

We end this section by providing an alternative expression of the LCP problem as the solution of a fixed point equation.

\subsubsection*{Alternative expression for the LCP solution} This fact will be useful in the next section.

\begin{proposition}\label{prop:fixed-point-LCP}
Let $\bs{z}=(z_i)\in \RR^N$ and consider the fixed-point equation:
\begin{equation}\label{eq:fixed-point-LCP}
\bs{z} = \Upsilon_N \bs{z}_+ +\bs{\rho}_N
\end{equation}
where $\bs{z}_+ =\left( (z_i)_+\right)$. Then $\bs{z}$ is a solution of \eqref{eq:fixed-point-LCP} iff $\bs{z}_+\in \LCP(I_N - \Upsilon_N, -\bs{\rho}_N)$. 
\end{proposition}

\begin{proof} Suppose that $\bs{z}$ is a solution of \eqref{eq:fixed-point-LCP} and write $\bs{z}=\bs{z}_+- \bs{z}_-$. Then
$$
\bs{z}_+, \bs{z}_- \succcurlyeq 0\,,\quad (\bs{z}_+)^\T \bs{z}_-=0\quad\textrm{and}\quad  
\bs{z}_-=(I_N- \Upsilon_N) \bs{z}_+ - \bs{\rho}_N\,. 
$$
Hence $\bs{z}_+ \in \LCP(I_N - \Upsilon_N, -\bs{\rho}_N)$. 

To establish the converse, let $(\bs{y},\bs{w})$ a solution of $\LCP(I_N - \Upsilon_N, -\bs{\rho}_N)$. Define $\bs{z}= \bs{y} - \bs{w}$ then 
$$
\begin{cases}
\bs{z}_+ = \bs{y}\\
\bs{z}_-=\bs{w}
\end{cases}
\qquad \textrm{and}\qquad 
\bs{w}=(I_N - \Upsilon_N) \bs{y} - \bs{\rho}_N \quad \Rightarrow \quad \bs{z}= \Upsilon_N \bs{z}_+ +\bs{\rho}_N\, .
$$
\end{proof}

\subsection{Design of an AMP algorithm to approximate the LCP solution} 
\label{subsec:GOE-algo}

\subsubsection*{The AMP principles in a nutshell}
We begin with some of the fundamental results of the AMP theory. 
The now classical form of an
AMP iterative algorithm, as formalized in the article \cite{bay-mon-11} of
Bayati and Montanari based in part on a result of Bolthausen~\cite{bol-14}, 
can be presented as follows. Let
$(h^k)_{k\ge 0}$ be a sequence of Lipschitz $\RR^2 \to \RR$ functions. By the Lipschitz assumption, the derivative 
$$\frac{\partial h^k(u,a)}{\partial u}$$ 
is defined almost everywhere and the function 
$\partial_1 h^k(u, a)$ is any function that coincides with this derivative 
where it is defined. For $\bs{x} = (x_i)_{i\in [N]}$, define by $\ps{\bs{x}}_N$ the scalar quantity:
$$
\ps{\bs{x}}_N := \frac 1N \sum_{i\in [N]} x_i\, .
$$ 
% and more generally $\ps{\bs{y}}_M := \frac 1M \sum_{i\in [M]} y_i$ for $\bs{y} = (y_i)_{i\in [M]}$ a $M\times 1$ vector.

Let $\bs{a}_N \in \RR^N$ be a random vector of so-called auxiliary information. Recall that 
$A_N$ is the GOE matrix introduced in Assumption \ref{ass:goe}. Starting with a 
vector $\bs{u}^0_N \in \RR^N$, the AMP recursion is written  
\begin{equation} 
\label{amp-goe} 
\bs{u}^{k+1}_N = \frac{A_N}{\sqrt{N}} h^k(\bs{u}^k_N, \bs{a}_N) \ - \ 
 \ps{\partial_1 h^k(\bs{u}^k_N, \bs{a}_N)}_N \ h^{k-1}(\bs{u}^{k-1}_N, \bs{a}_N)\,,
\end{equation} 
where $h^k(\bs{u},\bs{a}) = \left(h^k(u_i, a_i)\right)_{i\in [N]}$. 

From this recursion, it is possible to precisely evaluate the
asymptotic behavior of the empirical measures 
$$\mu^{\bs{a}_N, \bs{u}^1_N,
\cdots, \bs{u}^k_N}$$ as $N \to\infty$ for any $k$, and to prove that $\mu^{\bs{a}_N, \bs{u}^1_N,
\cdots, \bs{u}^k_N}$ converges toward a centered Gaussian vector whose covariance structure is defined by the so-called Density Evolution (DE). The term 
$$
 \ps{\partial_1 h^k(\bs{u}^k_N, \bs{a}_N)}_N \ h^{k-1}(\bs{u}^{k-1}_N, \bs{a}_N)
$$
(equal to zero for $k = 0$) is referred to as the Onsager term and plays a crucial role in making possible this convergence. For a detailed exposition of the AMP theory,
along with the description of many of its applications, the reader is referred
to the recent tutorial \cite{fen-etal-(now)22}.

\subsubsection*{A specific AMP algorithm for the LCP}
To establish Theorem~\ref{th:main-wigner}, we design the following
AMP algorithm and study its properties. For each $N$, let $(\bs{u}^0_N, \bs{a}_N) \in \RR^N \times \RR^N$ be a
couple of random vectors independent of $A_N$, with 
$\bs{a}_N \succcurlyeq 0$. Assume that there exists a couple of $L^2$ random 
variables $(\bar u, \bar a)$ such that 
\begin{equation}
\label{u0a} 
(a.s.) \qquad \mu^{\bs{u}^0_N, \bs{a}_N} \quad \xrightarrow[N\to\infty]{\cP_2(\RR^2)} \quad
 \mcL((\bar u, \bar a))\ ,\quad  \bar a \neq 0\, .
\end{equation} 
Vectors $\bs{u}^0_N$ and 
$\bs{a}_N$ will be specified later, see \eqref{eq:def-aN}. Notice that $\bar a \geq 0$. By Assumption \ref{ass:kappa}, $\kappa$ is larger than $\sqrt{2}$ hence \eqref{sys} admits an unique solution $(\delta, \sigma^2, \gamma)$ by the first part of the theorem. Let
$h^t \equiv h$ 
for all $k\ge 0$, where 
\[
h(u,a) = \frac{(u+a)_+}{\delta}\qquad\textrm{and}\qquad \partial_1 h(u,a) = \frac{\1_{\{u+a > 0\}}}{\delta}\, .
\]
The AMP iteration \ref{amp-goe} now reads  
\begin{equation}
\label{amp} 
\bs{u}^{k+1}_N = \frac{A_N}{\delta\sqrt{N}} \left(\bs{u}_N^k + \bs{a}_N\right)_+ - 
 \frac{\ps{\1_{\{\bs{u}^k_N + \bs{a}_N > 0\} }}_N \, \left(\bs{u}_N^{k-1} + \bs{a}_N\right)_+}{\delta^2} . 
\end{equation} 
The DE equations for this algorithm are provided by the following 
proposition, which is a direct application of
\cite[Theorem~2.3]{fen-etal-(now)22} (see also \cite[Theorem~4]{bay-mon-11}):  
\begin{proposition} 
\label{prop:amp-de} 
For $N\ge 1$, Let $A_N$ be a GOE matrix and let $(\bs{u}^0_N, \bs{a}_N) \in \RR^N \times \RR^N$ be a
couple of random vectors independent of $A_N$, with 
$\bs{a}_N \succcurlyeq 0$.
Assume \eqref{u0a} and consider the recursion \eqref{amp}. Then, for every $k \geq 1$, 
\[
(a.s.)\qquad \mu^{\bs{a}_N, \bs{u}^1_N, \cdots, \bs{u}^k_N} \quad 
 \xrightarrow[N\to\infty]{\cP_2(\RR^{k+1})} \quad 
 \mcL((\bar a, Z^1, \ldots, Z^k))\,,
\]
where $\left( Z^1 , \cdots, Z^{k} \right)$ is a centered Gaussian vector,
independent of $(\bar u, \bar a)$. The $ k\times k$ covariance matrix $R^{k}$ of the random vector $\left( Z^1 , \cdots, Z^{k} \right)$   is 
 defined recursively in $k$ as follows: 
$$
R^1 =\EE (Z^1)^2 = \frac{1}{\delta^2} \EE (\bar u + \bar a)_+^2\,,
$$
and given $R^k$, matrix $R^{k+1}$'s first principal submatrix is $R^k$,  
$$
\left[ R^{k+1}\right]_{ij} = \left[ R^{k}\right]_{ij} \qquad \textrm{for} \quad i,j\in [k]\,,
$$
whereas the last row and column of $R^{k+1}$ are defined via the equations:
\[
\left[ R^{k+1} \right]_{k+1,\ell} =  
\EE Z^{k+1} Z^\ell  = \frac{1}{\delta^2} 
 \left\{\begin{array}{lcl} 
  \EE (Z^k + \bar a)_+ (Z^{\ell-1} + \bar a)_+ &\text{if}& 
      \ell \in \{2,\ldots, k+1\} \,, \phantom{\bigg|}\\
  \EE (Z^k + \bar a)_+ (\bar u + \bar a)_+ &\text{if}& \ell = 1\,.  \phantom{\bigg|}
  \end{array}\right. 
\]
\end{proposition}

Notice that by writing $\bs{\alpha}^{k+1} = \left( (\bar u + \bar a)_+, 
  (Z^1 + \bar a)_+, \cdots, (Z^k + \bar a)_+ \right)^\T$, we see 
that $R^{k+1} = \EE \bs{\alpha}^{k+1} (\bs{\alpha}^{k+1})^\T$, which immediately shows that  
$ R^{k+1}$ is a positive semidefinite matrix (actually, one can prove that it is definite, see  
\cite{fen-etal-(now)22}).

Denote by $$\bs{\xi}_N^k = \bs{u}_N^k +\bs{a}_N\, .
$$
What is going to drive the following computations 
is the fact that the vectors $\bs{\xi}_N^k$ and $\bs{\xi}_N^{k+1}$ will tend to be aligned as $N\to \infty$ then $k\to \infty$. This will be formalized and proved in Lemma \ref{lemma:aligned}.
Denote by $\gamma^k_N = \big\langle \1_{\left\{ \bs{\xi}_N^k>0 \right\} }\big \rangle_N$ 
and recall the expression of $\gamma$ given 
in~\eqref{sys-gamma}. With these notations at hand, the AMP recursion \eqref{amp} reads:
\begin{eqnarray*} 
\bs{\xi}_N^{k+1} &=&  \frac{A_N}{\delta\sqrt{N}} \big(\bs{\xi}_N^k\big)_+ 
   - \frac{\gamma^k_N}{\delta^2}  \big(\bs{\xi}_N^{k-1}\big)_+ + \bs{a}_N\ ,  \\
&=&  \frac{A_N}{\delta\sqrt{N}} \big(\bs{\xi}_N^k\big)_+ 
   - \frac{\gamma}{\delta^2}  \big(\bs{\xi}_N^{k-1}\big)_+ + \bs{a}_N  +  \frac{\gamma - \gamma^k_N}{\delta^2} \big(\bs{\xi}_N^{k-1}\big)_+ \ ,\\
&=&     \frac{A_N}{\delta\sqrt{N}} \big(\bs{\xi}_N^k\big)_+ 
   - \frac{\gamma}{\delta^2}  \big(\bs{\xi}_N^{k}\big)_+ + \bs{a}_N  +  \frac{\gamma - \gamma^k_N}{\delta^2} \big(\bs{\xi}_N^{k-1}\big)_+ + \frac{\gamma}{\delta^2}  \left( \big(\bs{\xi}_N^{k}\big)_+ - \big(\bs{\xi}_N^{k-1}\big)_+\right) \ .\\
\end{eqnarray*}
Replacing now $\bs{\xi}_N^{k+1}$ by $\bs{\xi}_N^{k}$, we end up with:
\begin{equation}\label{eq:rough-LCP}
\bs{\xi}^k_N =  \frac{A_N}{\delta\sqrt{N}} \big(\bs{\xi}_N^k\big)_+ 
   - \frac{\gamma}{\delta^2}  \big(\bs{\xi}_N^k\big)_+ + \bs{a}_N + \bs{\varepsilon}^k_N, 
\end{equation}
where 
\begin{equation}
\label{eq:err} 
\bs{\varepsilon}^k_N = 
   \frac{\gamma - \gamma^k_N}{\delta^2} \big(\bs{\xi}_N^{k-1}\big)_+
  +\bs{\xi}_N^{k} - \bs{\xi}_N^{k+1}
  + \frac{\gamma}{\delta^2} \left( \big(\bs{\xi}_N^{k}\big)_+  - \big(\bs{\xi}_N^{k-1}\big)_+ \right)  . 
\end{equation} 
Massaging \eqref{eq:rough-LCP} and relying on \eqref{sys-kappa} we obtain:
\begin{equation}\label{eq:LCP-almost}
\big(\bs{\xi}_N^k\big)_+ - \frac{\big(\bs{\xi}_N^k\big)_-}{1+\gamma/\delta^2} = \frac{A_N}{\kappa\sqrt{N}} \big(\bs{\xi}_N^k\big)_+ +\frac{\delta(\bs{a}_N +\bs{\varepsilon}^k_N)}{\kappa}\, .
\end{equation}
Denote by
$$
\bs{z} = \big(\bs{\xi}_N^k\big)_+ - \frac{\big(\bs{\xi}_N^k\big)_-}{1+\gamma/\delta^2}\, .
$$
Notice that $\bs{z}_+ = \big(\bs{\xi}_N^k\big)_+$ and set finally 
\begin{equation}\label{eq:def-aN}
\bs{u}_N^0=\bs{1}_N\qquad \textrm{and}\qquad \bs{a}_N =\frac{\kappa}{\delta}\bs{r}_N\,.
\end{equation}
With these notations, \eqref{eq:LCP-almost} is rewritten 
\begin{equation}\label{eq:LCP-almost-bis}
\bs{z} = \Sigma_N \bs{z}_+ + \bs{r}_N  + \frac{\delta}{\kappa} \bs{\varepsilon}_N^k\, .
\end{equation}
Relying on Proposition \ref{prop:fixed-point-LCP} and on the fact that $\| \Sigma_N\|<1$ eventually, we conclude that $\bs{z}_+ = \big(\bs{\xi}_N^k\big)_+$ is the unique solution of 
$$
\LCP\left(I_N- \Sigma_N, -\bs{r}_N  - \frac{\delta}{\kappa} \bs{\varepsilon}_N^k\right)
$$ 
for $N$ large enough, which is almost what is aimed, up to the term $\frac{\delta}{\kappa} \bs{\varepsilon}_N^k$ - see Eq. \eqref{x*lcp}.

\begin{remark}
  Retrospectively, notice that with the choice \eqref{eq:def-aN}, assumptions of Proposition \ref{prop:amp-de} are satisfied: $(\bs{u}_N^0,\bs{a}_N)$ is independent of $A_N$ and \eqref{u0a} holds thanks to Assumption~\ref{ass:r} with $\bar a = \frac \kappa \delta \bar r$.     
\end{remark} 

Before bounding $\bs{\varepsilon}_N^k$, let us first study the behavior of $\mu^{\big(\bs{\xi}_N^k\big)_+}$. Applying Proposition \ref{prop:amp-de}, we get that for all $k\ge 2$:
$$
\mu^{\bs{u}^k_N} \xrightarrow[N\to\infty]{\cP_2(\RR)} \mcL(Z^k)\,,
$$
where $Z^k\eqlaw \theta_k \bar Z$ with $\bar Z \eqlaw {\mathcal N}(0,1)$ and $\theta_k$ satisfying the following DE equation:
\begin{equation}\label{eq:DE-intermediaire}
\theta^2_{k+1} = \frac{1}{\delta^2} \EE (\theta_k \oZ + \bar a )_+^2\, .
\end{equation}
Since function $\varphi(u,a) = (u+a)_+$ is Lipschitz, it is clear that
\begin{equation}\label{eq:conv-mu-plus}
\mu^{(\bs{\xi}^k_N)_+} \quad \xrightarrow[N\to\infty]{\cP_2(\RR)} \quad \mcL\left((\theta_k \oZ + 
 \bar a)_+ \right)\, .
\end{equation}
Furthermore, since the distribution function of 
$\theta_k \oZ + \bar a$ has no discontinuity, the following convergence holds:
$$
(a.s.) \qquad \gamma^k_N \xrightarrow[N\to\infty]{} 
   \PP\left( \theta_k \oZ +  \bar a > 0 \right)\qquad \textrm{where}\qquad \gamma^k_N = \big\langle \1_{\left\{ \bs{\xi}_N^k>0 \right\} }\big \rangle_N\, . 
$$
Introduce the quantity:
\begin{equation}\label{eq:lien-sigma-theta}
\sigma_k = \frac{\delta}{\kappa} \theta_k\, .
\end{equation}
Following \eqref{eq:DE-intermediaire}, the recursive equation satisfied by $\sigma_k$ is
$$
\sigma_{k+1}^2 =\frac 1{\delta^2} \mathbb{E} \left( \sigma_k \oZ +\bar r\right)_+^2
$$
which is precisely the equation appearing in Lemma \ref{lemma:point-fixe}-(ii). As a conclusion, $\sigma_k\xrightarrow[k\to\infty]{} \sigma$, where $\sigma$ satisfies \eqref{sys-omega}. This convergence has two interesting consequences:
$$
\PP\left( \theta_k \oZ +  \bar a > 0 \right) = \PP\left( \sigma_k \oZ +  \bar r > 0 \right)
\quad \xrightarrow[k\to\infty]{}\quad  \PP\left( \sigma \oZ +  \bar r > 0 \right) = \gamma\,,
$$
where $\gamma$ satisfies \eqref{sys-gamma}, and
$$
{\mathcal L}\left( (\theta_k \oZ +\bar a)_+ \right) = {\mathcal L}\left( \left( 1+\gamma/\delta^2\right) (\sigma_k \oZ +\bar r)_+ \right)\quad \xrightarrow[k\to\infty]{\cP_2(\RR)}\quad {\mathcal L}\left( \left( 1+\gamma/\delta^2\right) (\sigma \oZ +\bar r)_+ \right)\ ,
$$
the latter being the distribution appearing in Theorem \ref{th:main-wigner}-(iii).

\subsubsection*{Control of the error term $\bs{\varepsilon}^k_N$} 
Recall the expression of $\bs{\varepsilon}^k_N$ given in \eqref{eq:err}: 
$$
\bs{\varepsilon}^k_N = 
   \frac{\gamma - \gamma^k_N}{\delta^2} \big(\bs{\xi}_N^{k-1}\big)_+
  +\bs{\xi}_N^{k} - \bs{\xi}_N^{k+1}
  + \frac{\gamma}{\delta^2} \left( \big(\bs{\xi}_N^{k}\big)_+  - \big(\bs{\xi}_N^{k-1}\big)_+ \right) \, . 
$$ 
A direct consequence of \eqref{eq:conv-mu-plus} yields that 
$$
\frac{\| (\bs{\xi}_N^{k-1})_+ \|^2}{N} \xrightarrow[N\to\infty]{a.s.}\EE\left(\theta_{k-1} \oZ + 
 \bar a\right)_+^2 = \theta_k^2 \delta^2\,.
 $$
In particular, the sequence $\left(\frac{\| (\bs{\xi}_N^{k-1})_+ \|^2}{N} \right)_N$ is 
bounded. Furthermore, 
$\lim_k (a.s.) \lim_N (\gamma - \gamma^k_N) = 0$. We thus have 
\begin{equation}
\label{gg} 
\lim_{k\to\infty} (a.s.)\lim_{N\to\infty} 
   \frac{(\gamma - \gamma^k_N)^2}{\delta^4} \frac{\| (\bs{\xi}_N^{k-1})_+ \|^2}
 {N} = 0 \,.
\end{equation} 

The main idea to control the two remaining terms $\bs{\xi}_N^{k} - \bs{\xi}_N^{k+1}$ and $\big(\bs{\xi}_N^{k}\big)_+ - \big(\bs{\xi}_N^{k-1}\big)_+$ is to establish that the correlation coefficient 
\begin{equation}\label{eq:corr-coefficient}
  Q_k:=\frac{\EE Z^{k-1} Z^{k}} {\theta_{k-1} \theta_k}
\end{equation}
converges to $1$ as $k\to\infty$. This can be interpreted as an alignement of vectors $\bs{\xi}_N^k$ and $\bs{\xi}_N^{k-1}$. This argument was developed in a similar context in \cite{mon-ric-16}, see also 
\cite{don-mon-16}. For self-containedness, we state and prove the following lemma:
\begin{lemma}
\label{lemma:aligned} 
The sequence $(Q_k)_{k\ge 2}$ defined in \eqref{eq:corr-coefficient} 
satisfies $Q_k \xrightarrow[k\to\infty]{} 1$. 
\end{lemma} 
Proof of Lemma \ref{lemma:aligned} is postponed to Appendix \ref{proof:aligned}.

\vspace{0.25cm}

We now conclude the proof of Theorem \ref{th:main-wigner}.
Consider $\varphi(x_1, x_2) = (x_1 - x_2)^2 \in \PL_2(\mathbb R^2)$. By Proposition \ref{prop:amp-de}, we have
\[
(a.s.) \quad \frac{\| \bs{\xi}^k_N - \bs{\xi}^{k+1}_N \|^2}{N} = \frac 1N \sum_{i=1}^N 
 \varphi(u^k_i , u^{k+1}_{i}) 
 \xrightarrow[N\to\infty]{} \EE \left( Z^{k+1} - Z^k \right)^2 
  = \theta_{k+1}^2 + \theta_k^2 - 2\theta_{k+1} \theta_k Q_{k+1}\, .
\]
Applying Lemma \ref{lemma:aligned}, we get that:
\begin{equation}
\label{xx} 
\lim_{k\to\infty} (a.s.)\lim_{N\to\infty} \frac{\|  \bs{\xi}^k_N - \bs{\xi}^{k+1}_N \|^2}{N} 
 = 0\, .
\end{equation} 
A similar argument applies to the last term. 
\begin{multline*}
\frac 1N \| \big(\bs{\xi}^k_N\big)_+ - \big(\bs{\xi}^{k-1}_N\big)_+ \|^2 = \frac 1N 
\| \big(\bs{u}^k_N +\bs{a}_N\big)_+ - \big(\bs{u}^{k-1}_N+\bs{a}_N\big)_+ \|^2
\\
\xrightarrow[N\to\infty]{a.s.} \quad \mathbb{E} \left( (Z^k+\bar a)_+ - (Z^{k-1} +\bar a)_+\right)^2 
= \EE \left( Z^{k+1} - Z^k \right)^2 \, .
\end{multline*}
Finally, using that 
$$
  \frac{\|\bs{\varepsilon}^k_{N}\|^2}{N} \quad \leq\quad  
\frac 3N  \left( 
\frac{(\gamma - \gamma^k_N)^2}{\delta^4} \|(\bs{\xi}_N^{k-1})_+\|^2
+ \| \bs{\xi}_N^{k} - \bs{\xi}_N^{k+1} \|^2
+ \frac{\gamma^2}{\delta^4} \| \big( \bs{\xi}_N^{k}\big) _+  - (\bs{\xi}_N^{k-1})_+ \|^2
  \right),  
$$
we conclude that 
\begin{equation}
\label{eps} 
\lim_{k\to\infty} (a.s.)\lim_{N\to\infty} \frac{\|\bs{\varepsilon}^k_{N}\|^2}{N} 
 = 0\,.
\end{equation} 
Notice that the fact that the a.s. $\lim_N$ at the left hand side exists can be deduced again from 
Proposition~\ref{prop:amp-de}. %\jamal{je serais favorable à virer cette phrase}

\subsubsection*{From the approximated LCP to the genuine LCP}
%\subsubsection*{Use of a LCP perturbation result} 
Recall that whenever $\| \Sigma_N \| < 1$, which happens eventually, 
$$
\bs{x}_N^\star=\LCP (I_N-\Sigma_N,-\bs{r}_N)\qquad \textrm{and}\qquad \big(\bs{\xi}_N^k\big)_+= \LCP\left(I_N- \Sigma_N, -\bs{r}_N  - \frac{\delta}{\kappa} \bs{\varepsilon}_N^k\right)\,.
$$
Statistical properties have been established for $\big(\bs{\xi}_N^k\big)_+$ via the AMP procedure, see  
for instance \eqref{eq:conv-mu-plus}. Using LCP perturbation results, we shall identify the limiting empirical distribution of $\bs{x}_N^\star$. Let us introduce:
$$
\mu^\star ={\mathcal L}\left( ( 1+\gamma/\delta^2)(\sigma \oZ + \bar r)_+ \right) =
{\mathcal L}\left( \frac\kappa\delta (\sigma \oZ + \bar r)_+ \right)
$$
In \cite[Th.~2.7, Th.~2.8]{che-xia-07}, Chen and Xiang provide the following bound:
\begin{multline}\label{eq:bound-LCP}
\| \bs{x}^\star_N - (\bs{\xi}^k_N)_+ \| \ \leq\
 \left\| \left( I_N - \Sigma_N \right)^{-1} \right\| \times \frac \kappa \delta \left\| \bs{\varepsilon}^k_N \right\| \ = \ b_{N} \left\| \bs{\varepsilon}^k_N \right\|\\
 \textrm{where}\quad 
 b_N := \left\| \left( I_N - \Sigma_N \right)^{-1} \right\| \times \frac \kappa \delta\, .
\end{multline}
Let $\varphi : \RR\to\RR$ be an arbitrary function in $\PL(\mathbb R^2)$ with Lipschitz constant $L_\varphi$.
For a given positive integer $k$, we have 
\begin{align*} 
\frac 1N \sum_{i=1}^N \varphi(x^\star_{i}) - \int \varphi d\mu^\star 
 &= 
\frac 1N \sum_{i=1}^N 
  \left( \varphi(x^\star_{i}) - \varphi((\xi^k_i)_+) \right) 
 + \frac 1N \sum_{i=1}^N \varphi((\xi^k_i)_+) - 
    \int \varphi d\mu^\star \\
 &:= \epsilon^1_N(k) + \epsilon^2_N(k) .
\end{align*}
We first handle $\epsilon^2_N(k)$. By Proposition~\ref{prop:amp-de}, we have:
\[
\epsilon_N^2(k) \toaslong 
  \EE\, \varphi\left( \frac \kappa \delta (\sigma_k \oZ + \bar r)_+\right) - \EE\, \varphi\left( \frac \kappa \delta (\sigma \oZ + \bar r)_+\right)\, .
\]
The r.h.s. is easily bounded by a constant $C(k)$ which converges to zero as $k\to\infty$, using the fact that $\lim_k \sigma_k =\sigma$.

We now turn to $\epsilon^1_N(k)$. By Cauchy-Schwarz inequality  
\begin{eqnarray*}
\frac 1N \sum_{i=1}^N 
  \left| \varphi(x^\star_i) - \varphi((\xi^k_i)_+) \right| 
 &\le &  \frac {L_\varphi}N \sum_{i\in [N]}
  \left| x^\star_i - (\xi^k_i)_+\right| 
  \left( 1 + | x^\star_i | +  | (\xi^k_i)_+ |   \right)  \\ 
&\leq& \frac {L_\varphi}N
  \left\| \bs{x}^\star_{N} - (\bs{\xi}^k_{N})_+\right\| 
  \left( \sum_{i\in [N]} ( 1 + | x^\star_i | + | \xi^k_i)_+ |)^2  \right)^{1/2} 
   \\ 
&\leq &3L_\varphi 
  \frac{\left\| \bs{x}^\star_{N} - (\bs{\xi}^k_{N})_+\right\|}{\sqrt{N}}  
  \left( 1 + \frac{\| \bs{x}^\star_{N} \|}{\sqrt{N}} +  
   \frac{\| (\bs{\xi}^k_{N})_+ \|}{\sqrt{N}}  \right) . 
\end{eqnarray*}
Recall the bound \eqref{eq:bound-LCP} and the definition of $b_N$, then 
\[
| \epsilon^1_N(k) | \leq 3L_\varphi b_N \frac{\| \bs{\varepsilon}^k_N \|}{\sqrt{N}} 
 \left( 1 + 2 \frac{\| (\bs{\xi}^k_N)_+\|}{\sqrt{N}} +  
  b_N \frac{\| \bs{\varepsilon}^k_N \|}{\sqrt{N}} \right) . 
\]
By Assumption \ref{ass:kappa}, $b_N$ a.s. converges 
to a positive constant. By Proposition~\ref{prop:amp-de}, we furthermore have 
\[
\frac{\| (\bs{\xi}^k_N)_+\|}{\sqrt{N}} \toaslong 
 \left( \EE (\theta_k \oZ + \bar a)_+^2 \right)^{1/2}, 
\]
which is bounded in $k$. Using~\eqref{eps}, we obtain that 
$\limsup_N | \epsilon^1_N(k) |$ is bounded with probability one by a constant 
$C_1(k)$ which converges to zero as $k \to\infty$. Finally,
\[
(a.s.)\qquad \limsup_N\left| 
\frac 1N \sum_{i\in [N]} \varphi(x^\star_{i}) - \int \varphi d\mu^\star 
 \right| \leq C(k) + C_1(k) \, .
\]
Since $C(k) + C_1(k)$ can be made arbitrarily small, we have 
\[
(a.s.)\qquad  \frac 1N \sum_{i=1}^N \varphi(x^\star_{i}) 
 \xrightarrow[N\to\infty]{} \int \varphi d\mu^\star , 
\]
which ends the proof of Theorem~\ref{th:main-wigner}.

\section{Elements of proof of Theorem~\ref{th:wishart-main}} 
	\label{sec:proof-wishart}
The strategy of proof is similar to that of Theorem \ref{th:main-wigner}. The Wishart model induces differences for the design of the AMP algorithm that we describe hereafter. The full mathematical proof is a matter of careful bookkeeping of Section \ref{sec:proof-GOE}. We provide the main steps of the proof but skip many mathematical details which can be found in \cite{imane-phd}.	
\subsection{Existence and uniqueness of the solution of system \eqref{eq:sys-wishart}} This can be established as in the case of the GOE model with minor modifications and is hence skipped.  

\subsection{Design of an AMP algorithm to approximate the LCP solution}
We shall rely on the framework of asymmetric AMP as presented in \cite[Section 2.2]{fen-etal-(now)22}. Suppose that for a given $\kappa$ satisfying Assumption \ref{ass:kappa-wishart}, $(\delta,\tau^2,\gamma)$ is the unique solution of \eqref{eq:sys-wishart}. Consider the following recursive system:
\begin{subequations}
    \label{eq:recursion}
  \begin{eqnarray}
    \bs{u}^{k+1}_N &=& \frac{B_N^\T}{\sqrt{P}} \bs{v}_P^k - \frac{(\bs{u}^k_N +\bs{a}_N)_+}{\delta}\\
    \bs{v}_P^k &=& \frac{B_N}{\delta\sqrt{P}} (\bs{u}_N^k +\bs{a}_N)_+ - \frac NP \frac{\langle
    \bs{1}_{\{ \bs{u}_N^k + \bs{a}_N>0\}}\rangle_N}{\delta}  \bs{v}_P^{k-1}
    \end{eqnarray}
\end{subequations}

where $\bs{u}_N^k,\bs{u}_N^{k+1}$ are $N\times 1$ vectors and and $\bs{v}_P^{k-1}, \bs{v}_P^k$, $P\times 1$ vectors with initial conditions
$$\bs{u}_N^0=\bs{1}_N\qquad \textrm{and}\qquad \bs{v}_P^0 =\frac{B_N}{\delta \sqrt{P}}(\bs{u}_N^0 +\bs{a}_N)_+\ .
$$
The following proposition is the counterpart of Proposition \ref{prop:amp-de} for asymmetric AMP.

\begin{proposition}[consequence of Theorem~2.5 of \cite{fen-etal-(now)22}]
\label{prop-ampasym} 
    For $N,P\ge 1$, let Assumptions \ref{ass:Wishart-matrix}, \ref{ass:MP} and \ref{ass:kappa-wishart} hold true. Suppose that $\bs{a}_N \succcurlyeq 0$ is a random vector independent of $A_N$ satisfying
    $$
  (a.s.) \qquad \mu^{\bs{a}_N} \xrightarrow[\cP_2(\RR)]{N\to \infty} {\mathcal L}(\bar{a})
    $$
    and consider the recursions \eqref{eq:recursion}. Then for every fixed $k\ge 1$,
    \begin{eqnarray*}
    (a.s.)\qquad \mu^{\bs{a}_N, \bs{u}_N^1,\cdots, \bs{u}_N^k}& \xrightarrow[N,P\to\infty]{ \cP_2(\RR^{k+1})}& {\mathcal L}\left( (\bar a, U^1,\cdots, U^k) \right)\ ,\\
    (a.s.)\qquad \mu^{\bs{v}_N^0,\cdots, \bs{v}_N^{k-1}}& \xrightarrow[N,P\to\infty]{ \cP_2(\RR^{k})}& {\mathcal L}\left( (\bar a, V^0,\cdots, V^{k-1}) \right)\ ,\\
    \end{eqnarray*}
    where $(U^1,\cdots, U^k)$ is a centered Gaussian random vector independent of $\bar a$ with covariance $T^{[k]}$, and $(V^0,\cdots,V^{k-1})$ is a centered Gaussian random vector with covariance matrix $\Sigma^{[k]}$. More precisely the covariance matrices
    \begin{eqnarray*}
        T^{[k]} &=& (T_{ij})_{i,j\in [k]}\,; \qquad T_{ij} = \mathbb{E} U^i U^j\ ,\\
        \Sigma^{[k]} &=& (\Sigma_{i-1,j-1})_{i,j\in [k]}\,;\qquad \Sigma_{i-1,j-1} = \mathbb{E} V^{i-1} V^{j-1}\, 
    \end{eqnarray*}
    are defined inductively. First, let $\bar Z\sim {\mathcal N}(0,1)$ and introduce $\tau_k, \theta_k$ such that 
    $$
    V^k \eqlaw \theta_k \bar Z \qquad \textrm{and}\qquad U^k \eqlaw \tau_k \bar Z\, ,
    $$
    so that $\theta_k^2= \Sigma_{k,k}$ and $\tau_k^2=T_{kk}$. We define these quantities by induction:
    $$
    \theta_0^2 =\EE (1+\bar a)_+^2\,,\qquad \tau_{k+1}^2 =\mathbb{E} V_k^2= \theta_k^2\,,\qquad 
    \theta_{k+1}^2 =\frac c{\delta^2} \mathbb{E} (U_{k+1}+\bar a)_+^2\, .
    $$
    Now given $\Sigma^{[k]} = (\Sigma_{i-1,j-1})$, $\Sigma^{[k+1]}$ is defined by
    \begin{eqnarray*}
    \Sigma_{\ell,k} &=&\frac c{\delta^2} \mathbb{E} (U^\ell +\bar a)_+ (U^k +\bar a)_+\quad \textrm{for}\ \ell \in [k]\ ,\\
    \Sigma_{0,k} &=& \frac c{\delta^2} \mathbb{E} (1 +\bar a)_+ (U^k +\bar a)_+\ .
    \end{eqnarray*}
    Given $T^{[k]}= (T_{ij})$, $T^{[k+1]}$ is defined by
    \begin{eqnarray*}
        T_{\ell, k+1} &=& \mathbb{E} V^{\ell - 1} V^k \ =\ \Sigma_{\ell - 1,k}\quad \textrm{for}\ \ell \in [k+1]\ .\\
     \end{eqnarray*}
\end{proposition}

\subsubsection*{From AMP recursions to an approximate LCP solution} We introduce the following notations: 
$$
\bs{\xi}_N^k =\bs{u}_N^k +\bs{a}_N\ ,\qquad \gamma_N^k = \langle 1_{\{ \bs{\xi}_N^k>0\}} \rangle_N\, .
$$
Recall the definition of $\gamma$ solution to \eqref{eq:sys-wishart}. Performing similar computations as in Section \ref{subsec:GOE-algo}, we obtain:
\begin{equation}\label{eq:approx-FP-wishart}
\bs{\xi}_N^k +\frac{\big( \bs{\xi}_N^k\big)_+ }{\delta} = \frac{B_N^\T B_N}{\left( 1+\frac{c\gamma}{\delta}\right)\delta P} (\bs{\xi}_N^k)_+ +\bs{a}_N  + \widetilde{\bs{\varepsilon}}_N^k
\end{equation}
where 
$$
\widetilde{\bs{\varepsilon}}_N^k = \frac{B_N^\T }{\left( 1+\frac{c\gamma}{\delta}\right)\sqrt{P}}
\left( \frac{c\gamma - N/P \gamma_N^k}{\delta} \bs{v}_P^{k-1} +\frac{c\gamma}{\delta}\big( \bs{v}_P^k - \bs{v}_P^{k-1}\big) \right) +\bs{\xi}_N^k - \bs{\xi}_N^{k+1}\, .
$$
We introduce the following notations:
$$
\bs{z} = (\bs{\xi}_N^k)_+ - \frac{(\bs{\xi}_N^k)_-}{1+1/\delta}\ ,\quad \bs{r}_N =\frac{\bs{a}_N}{1+1/\delta}\ ,\quad \bs{\varepsilon}_N^k = \frac{\widetilde{\bs{\varepsilon}}_N^k}{1+1/\delta}\ .
$$
Then \eqref{eq:approx-FP-wishart} can be rewritten as 
$$
\bs{z} = \Sigma_N \bs{z}_+ + \bs{r}_N +\bs{\varepsilon}^k_N\, ,
$$
where $\Sigma_N$ is given by \eqref{eq:Wishart-matrix}. Applying Proposition \ref{prop:fixed-point-LCP}, we finally obtain that
$$
\bs{z}^+ =LCP\left( I_N - \Sigma_N, - \bs{r}_N - \bs{\varepsilon}_N^k\right)\, .
$$
The rest of the proof closely follows the corresponding part in the proof of Theorem \ref{th:main-wigner} and is omitted.

%%% OLD

\begin{appendix}
    \section{Theorem \ref{th:main-wigner}: remaining proofs}

    \subsection{Proof of Lemma \ref{lemma:point-fixe}}\label{proof:point-fixe}
Consider the function 
$f(\sigma^2) = \EE (\sigma \oZ + \bar r)_+^2$. Then, 
Equation~\eqref{sys-omega} is equivalent to the fixed-point equation:
\begin{equation}
\label{fixed} 
\frac{f(\sigma^2)}{\delta^2} = \sigma^2.
\end{equation}     
We can prove by elementary means that 
$$
\frac{df}{d\sigma^2} (\sigma^2) = \frac 1{2\sigma} \frac{df}{d\sigma}(\sigma^2) = \frac 1\sigma \mathbb{E}
\oZ (\sigma \oZ +\bar r)_+\ .
$$
Moreover, conditioning on $\bar{r}$ and applying the integration by parts formula for the Gaussian r.v. $\oZ$ we get
$$
\frac 1\sigma \mathbb{E} \left( \oZ (\sigma \oZ +\bar r)_+\mid \bar r\right) = 
\mathbb{E}\left(  1_{\{ \sigma \oZ +\bar r\ge 0\}} \mid \bar r\right)\, .
$$
Hence 
$$
\frac{df}{d\sigma^2} (\sigma^2) = \mathbb{P}\{ \sigma \oZ +\bar r\ge 0\} 
= \mathbb{P}\{ \oZ +\bar r/\sigma \ge 0\}\, .
$$
Notice that $\frac{df}{d\sigma^2} $ is a decreasing function since 
$$
\sigma <\sigma' \qquad \Rightarrow \qquad \{ \oZ +\bar r/\sigma' \ge 0\} \subset  
\{ \oZ +\bar r/\sigma \ge 0\}\,,
$$
with 
$$\lim_{\sigma^2\to \infty} \frac{df}{d\sigma^2} (\sigma^2) =\frac 12\,.
$$
We now introduce function $g(\sigma^2)=\frac{f(\sigma^2)}{\delta^2} - \sigma^2$. Notice that $g(0)= \mathbb{E} \bar r^2 /\delta^2>0$ and that 

\begin{equation}
\label{eq:deriv}
\frac{dg}{d\sigma^2} (\sigma^2) = \frac{\mathbb{P}\{ \oZ +\bar r/\sigma \ge 0\}}{\delta^2} - 1 > \frac 1{2\delta^2} -1\, .
\end{equation}

If $\frac 1{2\delta^2} -1\ge 0$ which is equivalent to the condition $\delta<(\sqrt{2})^{-1}$ then $g$'s derivative is positive hence $g$ is increasing with a positive starting point and never vanishes.

Suppose now that $\delta> 1/\sqrt{2}$. We shall prove that $g$ vanishes at a unique point $\sigma^2(\delta)$:
\begin{equation}\label{eq:g-single-point}
g(\sigma^2(\delta))=0 \qquad \textrm{for}\qquad  \sigma^2(\delta)>0\, .
\end{equation}
Notice that the derivative $dg/d\sigma^2$ is decreasing with a negative limit at infinity $$\lim_{\sigma^2\to \infty} \frac{dg}{d\sigma^2}(\sigma^2)= \frac 1{2\delta^2} - 1<0\, .$$

Depending on the sign of the value of $dg/d\sigma^2$ at zero, either $g$ is constantly decreasing from the positive value $g(0)$ or $g$ is first increasing then eventually decreasing. We now prove that 
\begin{equation}\label{eq:g-infinity}
\lim_{\sigma^2\to \infty} g(\sigma^2) <0\, .
\end{equation}
This will yield \eqref{eq:g-single-point}. 
$$
\frac{g(\sigma^2)}{\sigma^2} = \frac{\mathbb{E} (\sigma \oZ +\bar r)^2_+ }{\delta^2 \sigma^2} - 1
=  \frac{\mathbb{E} (\oZ +\bar r/\sigma )^2_+ }{\delta^2} - 1 \xrightarrow[\sigma^2\to\infty]{} \frac 1{2\delta^2}-1<0\, .
$$
Hence $g$'s limit is $-\infty$ at infinity. Eq. \eqref{eq:g-infinity} is proved, so is \eqref{eq:g-single-point}. The first statement of the lemma is proved.\\

We now address the second point of the lemma. Let $\delta > 1/\sqrt 2$ be fixed.
From the previous analysis, we know that 
$$  \frac{dg}{d\sigma^2}\bigg|_{\sigma^2=\sigma^2(\delta)} <0.$$
From \eqref{eq:deriv}, one can compute
$$
 \frac{dg}{d\sigma^2}\bigg|_{\sigma^2=\sigma^2(\delta)} = \frac{\mathbb{P}\{ \sigma(\delta) \oZ +\bar r\ge 0\}}{\delta^2} - 1, $$
 and this gives the second point :
 $$\mathbb{P}\{\sigma(\delta) \oZ +\bar r\ge 0\} < \delta^2\, . $$

We now address the third point of the lemma. Consider a sequence $(\sigma_t)$ such that
$$  \sigma_0^2>0 \textrm{ and }
    \sigma_{p+1}^2 = \frac 1{\delta^2} f(\sigma_p^2)\,.
$$
One can easily  prove that $\sigma_p^2 \uparrow_p \sigma^2(\delta)$ (resp. $\sigma_p^2 \downarrow \sigma^2(\delta)$) if $\sigma_0^2< \sigma^2(\delta)$ (resp. $\sigma_0^2> \sigma^2(\delta)$). The sequence remains constant if $\sigma^2_0=\sigma^2(\delta)$. Lemma \ref{lemma:point-fixe} is proved.

\subsection{Proof of Lemma \ref{lemma:aligned}} \label{proof:aligned}
\begin{proof} 
Let $( X_1, X_2 )$ be a centered Gaussian vector with covariance matrix $\Gamma(X_1,X_2)$ given by
$$
\Gamma(X_1,X_2) = \begin{pmatrix}
1&q\\
q&1
\end{pmatrix}\qquad \textrm{with}\qquad q\in [0,1]\ . 
$$
Let $W$ be
a (real) random variable independent of $( X_1, X_2 )$ with finite second moment $\mathbb{E} W^2 <\infty.$ Consider the function 
${\mathcal H}: [0,1] \to [0,1]$ defined as 
\[
q\ \longmapsto \ {\mathcal H}(q) \ = \ \frac{\EE (X_1 + W)_+ (X_2 + W)_+}{\EE (X_1 + W)_+^2}\, . 
\]
It is shown in \cite[Lemma~38 and proof of Lemma~37]{mon-ric-16} that ${\mathcal H}$ is
a continuous increasing function on $[0,1]$ such that 
$$
{\mathcal H}(q) >q \quad \textrm{for all} \quad q<1 \qquad \textrm{and} \qquad {\mathcal H}(1) = 1\, .
$$
% Thus, starting with any $q \in [0,1)$, the sequence $q_{t+1} = \sH(q_t)$ is
% increasing and is bounded by $1$. Since its limit $q_\infty$ satisfies
% $q_\infty = \sH(q_\infty)$, it holds that $q_\infty = 1$. 

Let $Z^k$ be defined in Proposition \ref{prop:amp-de}, $\theta_k$ in \eqref{eq:DE-intermediaire} and $Q_k$ in \eqref{eq:corr-coefficient}. Writing $Z^k = \theta_k \bar Z^k$ where ${\mathcal L}\left( \bar{Z}^k\right) = {\mathcal N}(0,1)$, notice that 
$$
\mathrm{Cov}\left(\bar Z^k, \bar Z^{k-1}\right) = Q_k\ .
$$
We have 
\begin{eqnarray*}
Q_{k+1} &= &\frac{\EE Z^k Z^{k+1}}{\theta_k \theta_{k+1}} = 
 \frac{\EE (\theta_{k-1} \bar Z^{k-1} + \bar a)_+  
           (\theta_{k} \bar Z^{k} + \bar a)_+} 
   {\sqrt{\EE (\theta_{k-1} \bar Z^{k-1} + \bar a)_+^2 
    \EE (\theta_{k} \bar Z^{k} + \bar a)_+^2   }} \ ,  \\
&=& \frac{\EE (\bar Z^{k-1} + \bar a / \theta_{k-1})_+  
           (\bar Z^{k} + \bar a /\theta_k )_+} 
   {\sqrt{\EE (\bar Z^{k-1} + \bar a / \theta_{k-1})_+^2 
    \EE (\bar Z^{k} + \bar a /\theta_k)_+^2   }} \ .
\end{eqnarray*}
Notice that the last expression only depends on $\theta_{k-1}$, $\theta_k$ and $Q_k$, the covariance between $\bar Z^k$ and $\bar Z^{k-1}$. We thus introduce the following function
$$
 H(Q_k, \theta_{k-1}, \theta_k) \quad =\quad \frac{\EE (\bar Z^{k-1} + \bar a / \theta_{k-1})_+  
           (\bar Z^{k} + \bar a /\theta_k )_+} 
   {\sqrt{\EE (\bar Z^{k-1} + \bar a / \theta_{k-1})_+^2 
    \EE (\bar Z^{k} + \bar a /\theta_k)_+^2   }} \ .
$$
The function $H$ is continuous. Combining Eq. \eqref{eq:lien-sigma-theta} and the convergence of $\sigma_k$, denote by $\theta_\infty= \frac{\kappa}{\delta} \sigma$ where $\sigma$ satisfies \eqref{sys-omega}. If we set 
$W = \bar a / \theta_\infty$ in the definition of ${\mathcal H}$ above, then 
$$
{\mathcal H} (q) = H(q,\theta_\infty, \theta_\infty)\, .
$$ 
%Also, it is clear from the previous display that $Q_k \in [0,1]$. 
The lemma is established if we prove that $Q_\star := \liminf_k Q_k$
satisfies $Q_\star = 1$. Let us first show that $\liminf {\mathcal H}(Q_k) \geq
{\mathcal H}(Q_\star)$. If $Q_\star=0$, then $Q_k\ge Q_\star$ and since ${\mathcal H}$ is increasing we have  $\liminf {\mathcal H}(Q_k) \geq {\mathcal H}(Q_\star)$. It is thus sufficient to assume that $Q_\star > 0$.  For each
$\varepsilon > 0$, $Q_k \geq Q_\star - \varepsilon$ for all $k$ large enough.
Thus, ${\mathcal H}(Q_k) \geq {\mathcal H}(Q_\star - \varepsilon)$ for all $k$ large, which
implies that $\liminf {\mathcal H}(Q_k) \geq {\mathcal H}(Q_\star - \varepsilon)$. Since
$\varepsilon > 0$ is arbitrary, we have $\liminf {\mathcal H}(Q_k) \geq {\mathcal H}(Q_\star)$.
With this, we have 
\begin{eqnarray*}
Q_\star &=& \liminf_k H(Q_k, \theta_{k-1}, \theta_k) \ \stackrel{(a)}= \ \liminf_k H(Q_k, \theta_{\infty}, \theta_\infty)\ = \ \liminf_k {\mathcal H}(Q_k)\,, \\
 &\geq& {\mathcal H}(Q_\star) ,  
\end{eqnarray*} 
where $(a)$ follows from the continuity of $H$. By ${\mathcal H}$'s properties, this implies that $Q_\star = 1$. 
\end{proof} 
\section{Elements of proof for Theorems \ref{th:univ-goe} and \ref{th:asym-univ} (universality)}\label{app:universality}
We provide hereafter arguments to complete proofs of Theorems \ref{th:univ-goe} and \ref{th:asym-univ} based on what has already been developed in the proofs of Theorems \ref{th:main-wigner} and \ref{th:wishart-main} and on various results available in the literature. 

\begin{proof}[Proof of Theorem \ref{th:univ-goe}]
We just need to prove that Proposition~\ref{prop:amp-de} above remains true
when Assumptions~\ref{ass:goe} and~\ref{ass:r} are replaced with 
Assumptions~\ref{ass:univ-goe} and~\ref{ass:r-strong} respectively. 
This is a direct application of \cite[Theorem~2.4]{wan-zho-fan-(arxiv)22}. 
\end{proof} 

\begin{proof}[Proof of Theorem \ref{th:wishart-main}] 
We only need to prove that Proposition~\ref{prop-ampasym} remains true with the
assumptions of Theorem \ref{th:asym-univ}. To that end, it is enough to
notice that \cite[Theorem~2.5]{fen-etal-(now)22}, from which
Proposition~\ref{prop-ampasym} follows directly, can in turn be cast in the
framework of the AMP algorithm for GOE matrices~\eqref{amp-goe}, thanks to the
embedding of Javanmard and Montanari described in~\cite{jav-mon-13}.  Indeed,
Assumptions~\ref{ass:univ-wish} and~\ref{ass:r-strong} used in conjuction with
this embedding provide a version of Algorithm~\eqref{amp-goe} that enters the
framework of \cite[Theorem~2.4]{wan-zho-fan-(arxiv)22}. This leads to
Proposition~\ref{prop-ampasym}. 
\end{proof}

\end{appendix}

% \bibliographystyle{plain}
% \bibliography{math} 

\def\cprime{$'$} \def\cdprime{$''$} \def\cprime{$'$} \def\cprime{$'$}
  \def\cprime{$'$} \def\cprime{$'$}

\end{document}